\documentclass[a4paper,10pt]{article}
\pdfoutput=1 % if you are submitting a pdflatex (i.e. if you have
\usepackage{jheppub} % for details on the use of the package, please
\usepackage[export]{adjustbox}
\usepackage[utf8]{inputenc}
\usepackage{multirow}
\usepackage{bbold}
\usepackage[table]{xcolor}
\usepackage{slashed}
\usepackage{bm}
\usepackage{subfig}
\usepackage{float}
\usepackage{fancyvrb}
\usepackage{fvextra}
\usepackage{soul}
\usepackage{comment}
\usepackage{cancel} %Required for \cancel
\usepackage{csquotes} %Required for \enquote
\usepackage[title]{appendix}

\usepackage{amsmath,amssymb,amsthm}
\usepackage{mathtools}
\usepackage{enumitem}
\usepackage{hyperref}

\usepackage{subfig}
\usepackage{color,graphicx,slashed,hyperref}
\usepackage[utf8]{inputenc}
 \usepackage{verbatim}
\usepackage{setspace}
\usepackage{booktabs}
\usepackage{amstext} % for \text macro
\usepackage{array}   % for \newcolumntype macro
\newcolumntype{C}{>{$}c<{$}} % math-mode version of "l" column type
\newcolumntype{L}{>{$}l<{$}} % math-mode version of "l" column type

\newcommand{\Mpl}{M_{\rm Pl}}
\newcommand{\dd}{\partial}
\newcommand{\cS}{\mathcal S}

\theoremstyle{plain}
\newtheorem{theorem}{Theorem}
\newtheorem{lemma}{Lemma}
\newtheorem{proposition}{Proposition}[section]
\theoremstyle{definition}
\newtheorem{corollary}{Corollary}[section]
\newtheorem{remark}{Remark}

\title{Adiabatic Deformations of Black Hole Moduli: \\ I. Canonical Slices and Fredholm Solvability}
\author[a]{Karim Benakli,}
\emailAdd{kbenakli@lpthe.jussieu.fr}

\author[a]{Anna Chrysostomou}
\emailAdd{chrysostomou@lpthe.jussieu.fr}

\affiliation[a]{Laboratoire de Physique Th\'eorique et Hautes \'Energies - LPTHE, Sorbonne Universit\'e, CNRS, 4 Place Jussieu, 75005 Paris, France}

\abstract{
The leading adiabatic response of a static black hole to a slowly evolving asymptotic modulus separates into a geometric problem and a functional-analytic one: the former identifies the tangent direction on the manifold of static solutions actually followed by the black hole, and the latter determines whether the resulting deformation admits a unique solution. We show that both can be characterized by explicit scalar criteria. Differentiating a smooth family of static solutions defines an intrinsic covector on its tangent space. When the family is parameterized by the asymptotic modulus and the static horizon-mass parameter, the kernel of this covector identifies a canonical tangent direction that coincides with the one selected by a slowly varying exterior modulus. Promoting the static family to a slowly evolving instantaneous representative does not solve the exact field equations; the resulting discrepancy defines a lag field governed, after elimination of the metric perturbations, by a single reduced radial operator. 
A representation convention separates the exact dynamical horizon from that of the instantaneous representative. A local slice then removes the remaining tangential ambiguity and induces a finite-rank perturbation of the reduced operator. The Fredholm index is preserved, and, whenever the finite-rank update is nontrivial, the canonical tangent zero mode is displaced from the kernel of the reduced operator. Under the transversality and Fredholm hypotheses, existence and uniqueness of the completed leading-order adiabatic boundary-value problem reduce to a single scalar matching condition on the resulting kernel generator. The construction depends only on a smooth non-extremal static solution manifold and the Fredholm properties of the reduced operator. A companion paper specializes the framework to the magnetic GHS family.
}

\keywords{}

\begin{document}
\maketitle
%\flushbottom

%===================================================
%\newpage
\section{Introduction}
\label{sec:introduction}

Static black hole solutions can form smooth families parameterized by asymptotic moduli (such as scalar charges, dilaton values, or other continuous data fixed at infinity); Einstein--Maxwell--dilaton theories are examples thereof, where the scalar couples nontrivially to the gauge sector~\cite{Gibbons:1987ps,Garfinkle:1990qj}. Such couplings naturally admit continuous families of black hole solutions carrying nontrivial scalar profiles, lying outside the standard minimally coupled no-hair setting. When such a modulus varies slowly in time, the black hole cannot in general remain exactly static, and two distinct questions arise. The first is geometric: among the many possible instantaneous representatives of the exact black-hole solution within the static solution family, which representative is distinguished, at leading adiabatic order,
by the intrinsic geometry of the solution manifold? The second is functional-analytic: once this trajectory is identified, the black hole's leading-order adiabatic response is governed by a linear perturbation equation -- linear because only the leading order in the slow-roll parameter is kept -- whose reduced radial operator typically possesses a nontrivial kernel, inherited from the tangent directions of the static family itself, so that existence and uniqueness of the adiabatic deformation are not automatic. 

This paper answers both questions in turn for a general class of Einstein--Maxwell--dilaton static families, and in the process identifies the single scalar condition that determines whether the second question has a definite answer. It is organized around two complementary results: a geometric theorem identifying the canonical tangent direction on the static family, and a solvability theorem characterizing the existence and uniqueness of the corresponding adiabatic deformation. Note that this solvability theorem depends essentially on the geometric theorem, as the latter supplies the distinguished tangent direction whose fate under the sliced operator ultimately controls the solvability criterion. Thus, the existence and uniqueness of the leading-order adiabatic
boundary-value problem reduce entirely to the Fredholm solvability criterion.

%The static solutions of interest form a family $\cS$, parameterized by collective coordinates $\lambda^A$; each point of $\cS$ is itself an entire static field configuration $(\phi_0,\delta_0,m_0)(r;\lambda)$, so the tangent space $T_\lambda\cS$ at a point of $\cS$ consists of infinitesimal deformations of this configuration, obtained by differentiating the exact static field equations with respect to $\lambda^A$, and should not be confused with the tangent space of spacetime itself. 
The static solutions of interest form a family $\cS$, parameterized by collective coordinates $\lambda^A$. For each choice of $\lambda^A$, the corresponding point of $\cS$ is itself an entire static field configuration,
\begin{equation}
   \Psi_0(r;\lambda)
\equiv
\bigl(
\phi_0(r;\lambda),
\delta_0(r;\lambda),
m_0(r;\lambda)
\bigr) \;. 
\end{equation}

We write this compactly as $\Psi_0(r;\lambda)~=~(\phi_0,\delta_0,m_0)(r;\lambda)$. The tangent space $T_\lambda \cS$ therefore consists of infinitesimal deformations of these field profiles as the parameters $\lambda^A$ are varied. These tangent directions are obtained by differentiating the static solution with respect to the family parameters,
\begin{equation}
\partial_A\Psi_0 =
\bigl(
\partial_A\phi_0,
\partial_A\delta_0,
\partial_A m_0
\bigr) \;.    
\end{equation}

Note that $T_\lambda\mathcal S$ is the tangent space of the manifold of static solutions and should not be confused with the tangent space of spacetime. A generic element of $T_\lambda\cS$ need not lie in the homogeneous kernel of the radial operator obtained after reducing the linearized field equations: it typically sources a residual mass-integration constant associated with that tangent direction $A$. Collecting these residues defines an intrinsic linear functional on $T_\lambda\cS$, whose kernel 
\begin{equation}
\ker\mathcal C_\mu\subset T_\lambda\cS \;,
\end{equation}
is a coordinate-independent, canonically distinguished subspace. 

The choice of collective coordinates is crucial for identifying the
distinguished tangent direction selected by the intrinsic covector. We use the asymptotic modulus $\Phi$ together with the static horizon-mass parameter $M_{{\rm H},0}$ as local coordinates on the static solution manifold $\mathcal S$. The use of the horizon mass, rather than an asymptotic mass, is natural in the finite near-zone construction adopted throughout this paper, where no asymptotically flat region is included and no asymptotic mass need be defined.

Upon promotion of the static family, the coordinate
$M_{{\rm H},0}$ becomes the horizon-mass parameter
$M_{\rm H}^{\rm inst}(v)$ of the instantaneous static representative. This quantity should not be identified pointwise with the exact dynamical horizon mass $M_{\rm H}^{\rm exact}(v)$ of the full solution. The distinction between these two notions, and the corresponding lag-induced displacement of the exact horizon relative to the instantaneous representative, plays an essential role in the adiabatic construction developed later.

The horizon mass is singled out because the exact dynamical horizon
mass satisfies
\begin{equation}
 \dot M_{\rm H}^{\rm exact}
=
{\mathcal O}_{\rm ad}(\epsilon^2),   
\end{equation}
whereas generic near-zone quantities evolve already at first
adiabatic order. This makes $M_{{\rm H},0}$ the natural second
collective coordinate on the static solution manifold. Using the
coordinates $(\Phi,M_{{\rm H},0})$, we prove that the externally
driven tangent direction
\begin{equation}
 \left.
\partial_\Phi
\right|_{M_{{\rm H},0}},   
\end{equation}
namely the direction excited by a slowly varying exterior modulus,
belongs to the distinguished kernel introduced above. This follows
directly from the static identity
\begin{equation}
M_{{\rm H},0}
=
\frac{\Mpl^2}{2}\,
r_{{\rm H},0},   
\end{equation}
and is therefore a purely geometric property of the static solution
family, independent of any particular black hole solution.

The instantaneous-tracking configuration built from this tangent
direction does not itself solve the exact time-dependent field
equations: the resulting discrepancy defines a lag field, sourced at first adiabatic order, whose dynamics reduces -- after elimination of the metric perturbations -- to a single linear equation for a scalar radial function. The decomposition of the exact solution into an instantaneous representative and a lag field is not unique. After fixing the representation in the static horizon-mass direction, a single residual tangential ambiguity remains, corresponding to redistributing the solution along the canonical modulus direction. This ambiguity is removed by a local slice, which converts the remaining radial integration function into a boundary evaluation of the scalar lag field. The resulting sliced operator differs from the unsliced reduced operator by a finite-rank perturbation,
\begin{equation}
L_{\rm sl}
=
L
+
u\otimes\ell,    
\end{equation}
which plays the central role in the subsequent solvability analysis.

Finite-rank perturbations preserve the Fredholm index~\cite{Kato:1995}. In the present problem, this has a direct geometric consequence. The local slice removes the canonical tangent mode $\varphi_\Phi$ from the kernel of the reduced operator whenever the induced rank-one update is nontrivial. However, because the Fredholm index is unchanged, the one-dimensional kernel cannot simply disappear. Instead, it is relocated to a new generator $\psi_\star$ of the sliced operator $L_{\rm sl}$.

This relocated mode is not itself the physical adiabatic deformation. It spans the one-dimensional kernel of the sliced radial operator, so it represents the only remaining freedom once the local slice has separated the solution into an instantaneous static representative and a lag field. We denote by \(\psi_\star\) a nonzero generator of this kernel.

To recover the actual adiabatic deformation, this remaining freedom must still be fixed, and only the outer matching condition can determine its amplitude. We show that, once the local slice is transverse, existence and uniqueness of the leading-order adiabatic boundary-value problem reduce to the single scalar condition
\begin{equation}
\mathcal B_{\rm m}\psi_\star\neq0,
\end{equation}
where \(\mathcal B_{\rm m}\) is the bounded linear functional supplied by the outer matching problem. Equivalently, this condition states that the matching functional does not annihilate the remaining homogeneous mode, so that its amplitude is fixed uniquely by the exterior data. This Fredholm criterion is the main abstract result established in the present paper.

The response of a black hole to a slowly varying scalar background has been
studied from several complementary perspectives. Jacobson's construction of
a horizon-regular combination of scalar modes in advanced coordinates showed
the absence of gravitational memory for a linearly rolling modulus around
Schwarzschild~\cite{Jacobson:1999vr}. This result was subsequently used to
compute observable consequences, including dipole radiation from
unequal-mass binaries and constraints from the binary quasar
OJ287~\cite{Horbatsch:2011ye}. Related slow-roll constructions have also
been developed for black holes embedded in scalar-field cosmological
backgrounds, incorporating scalar accretion and its backreaction on the
geometry~\cite{Chadburn:2013mta}.

A separate line of work established no-scalar-hair results for static black
holes under appropriate assumptions on the scalar
theory~\cite{Bekenstein:1995un}. More recently, structural obstructions
based on shift symmetry, regularity, and conserved currents led to the
no-hair theorem for the galileon~\cite{Hui:2012qt}. The assumptions
underlying that theorem were subsequently shown not to cover
shift-symmetric scalar--tensor theories with a Gauss--Bonnet
coupling~\cite{Sotiriou:2013qea}, nor configurations in which the scalar
depends linearly on time while the geometry remains
static~\cite{Babichev:2013cya}; see Ref.~\cite{Volkov:2016ehx} for a
broader review of these and other hairy black-hole solutions.

These works address questions closely related to, but distinct from, the
one considered here. Existing no-hair theorems determine when static or
linearly time-dependent scalar configurations are obstructed, while
explicit rolling and slow-roll constructions describe particular
time-dependent black-hole solutions. The present work instead asks when a
slowly forced deformation of an entire static solution family exists and
is unique. The primary object is therefore the static solution manifold
itself, and the adiabatic problem is formulated as a deformation problem
on that manifold. Its existence and uniqueness are reduced to a single
explicit non-degeneracy condition. This viewpoint also underlies our
companion analysis of black-hole-induced moduli excursions in the
presence of a scalar potential~\cite{Benakli:2026zbm}, for which the
present construction provides the general geometric and
functional-analytic foundation.

Black-hole solution spaces have also been used as probes of moduli-space
geometry and related consistency constraints~\cite{Delgado:2022dkz},
while in the extremal setting the attractor mechanism identifies
distinguished critical points in moduli
space~\cite{Ferrara:1997tw}. The present construction addresses a
different question: it identifies the tangent direction selected by a
slowly time-dependent exterior modulus on a generally non-extremal static
solution manifold.

A distinct tradition, dating back to the moduli-space approximation for solitons~\cite{Manton:1981mp} and later applied to BPS and extremal black holes~\cite{Ferrell:1987gf,Gibbons:1986cq}, projects slow dynamics onto the tangent space of a static solution manifold, typically via an $L^2$ orthogonality condition fixed once a metric on that space is available. The construction developed here shares its starting point -- the static solution manifold and its tangent space -- but departs from this tradition in two respects: the slice is built locally, as a boundary condition at a finite matching radius, without presupposing an inner product on the solution space; and the resulting finite-rank modification of the reduced operator, together with the index-theoretic argument governing the fate of its kernel, is developed as an explicit functional-analytic construction rather than assumed to project cleanly onto geodesic motion.

The present work naturally separates into two layers. The first is
model dependent: starting from the Einstein--Maxwell--dilaton field
equations, the exact near-zone system is reduced to a single radial
scalar operator together with a finite-dimensional static solution
manifold. The second is independent of the detailed form of the
Einstein--Maxwell--dilaton action. Once the reduced operator and the
static solution manifold have been obtained, the canonical tangent
construction, the local slice, the finite-rank perturbation of the
reduced operator, and the resulting solvability theorem depend only on
their geometric and functional-analytic properties. They therefore
apply equally to any theory admitting a smooth non-extremal static
solution manifold together with a Fredholm reduced radial operator.

The purpose of the present paper is to isolate this general geometric
and functional structure. A companion paper
\cite{Benakli:2026PaperII} specializes the construction to the magnetic
Garfinkle--Horowitz--Strominger (GHS) family, where all geometric,
analytic, and matching ingredients can be evaluated explicitly.

%%%%%%%%%%%%%%%%%%%%%%%%%%%%%%%%%%%%%%%%%%%%%%%%%%%%%%%%%%%%%%%%%%%%%%%%

%%%%%%%%%%%%%%%%%%%%%%%%%%%%%%%%%%%%%%%%%%%%%%%%%%%%%%%%%%%%%%%%%%%%%%%%

\setcounter{section}{1}
\section{Exact reduction to the near-zone problem}
\label{sec:reduction}

We begin by formulating the problem in Einstein--Maxwell--dilaton
theory. The equations derived in this section are exact and do not rely
on an adiabatic expansion. Their purpose is to identify the finite
near-zone system from which the reduced scalar operator will be
constructed in Section~\ref{sec:operator_and_geometry}.

%%%%%%%%%%%%%%%%%%%%%%%%%%%%%%%%%%%%%%%%%%%%%%%%%%%%%%%%%%%%%

\subsection{Theory and exact field equations}
\label{subsec:theory}

We use metric signature $(-,+,+,+)$. The Ricci scalar and Einstein
tensor are denoted by $R$ and $G_{\mu\nu}$, respectively,
$\nabla_\mu$ is the spacetime covariant derivative, and
$\Box\equiv g^{\mu\nu}\nabla_\mu\nabla_\nu$. The reduced Planck mass
$\Mpl$ is fixed by the normalization of the Einstein--Hilbert term
below.

A comma subscript denotes a spacetime partial derivative, as in
$\phi_{,r}$ and $\phi_{,v}$, while a prime denotes differentiation
with respect to the scalar argument, as in $V'(\phi)$ and
$B'(\phi)$. A trailing $|_{\rm H}$ denotes evaluation on the future
outer marginal horizon defined in
Eq.~\eqref{eq:horizon-definition}. The Maxwell field strength is
denoted by $F_{\mu\nu}$, with
\begin{equation}
F^2\equiv F_{\mu\nu}F^{\mu\nu},    
\end{equation}
whereas $\mathcal F(v,r)$ is reserved for the metric function
introduced in Eq.~\eqref{eq:metric}.\footnote{The angular coordinates
are $(\theta,\varphi)$, with
$d\Omega_2^2\equiv d\theta^2+\sin^2\theta\,d\varphi^2$. The azimuthal
angle $\varphi$ should not be confused with the scalar field $\phi$.}

The dynamical scalar field is $\phi(v,r)$. The prescribed modulus
characterizing the slowly varying exterior is denoted by $\Phi(v)$.
It is not an additional dynamical field of the near-zone problem, but
an externally specified function of the advanced time normalized at
the matching surface. The relation between $\Phi(v)$ and the
near-zone scalar field at the finite matching radius is supplied by
the exterior matching problem; in particular, one should not assume
a priori that $\phi(v,r_{\rm m})=\Phi(v)$.

We consider Einstein--Maxwell--dilaton theory~\cite{Gibbons:1987ps},
\begin{equation}
S
=
\int d^4x\,\sqrt{-g}
\left[
\frac{\Mpl^2}{2}R
-\frac12(\nabla\phi)^2
-V(\phi)
-\frac14 B(\phi)F_{\mu\nu}F^{\mu\nu}
\right].
\label{eq:action}
\end{equation}
Here $V(\phi)$ is the scalar potential and $B(\phi)$ is a positive
gauge kinetic function. Unless stated otherwise, both functions are
kept arbitrary. We assume $V,B\in C^2$ on the range of scalar-field
values attained by the background family. The static background
profiles are likewise assumed sufficiently regular for the
differentiations and linearizations performed below; any stronger
regularity required by the functional-analytic construction will be
stated explicitly when needed.

With these conventions, the scalar field $\phi$ has mass dimension
one in four spacetime dimensions, while $B(\phi)$ is dimensionless.
The positivity assumption
\begin{equation}
B(\phi)>0
\label{eq:B-positive}
\end{equation}
ensures that the Maxwell field has the standard sign kinetic term.
No particular functional form of $B$ is required in the present
paper.

Variation of Eq.~\eqref{eq:action} gives
\begin{align}
\Mpl^2G_{\mu\nu}
&=
T_{\mu\nu}^{(\phi)}+T_{\mu\nu}^{(F)},
\label{eq:einstein-general}
\\
T_{\mu\nu}^{(\phi)}
&=
\nabla_\mu\phi\,\nabla_\nu\phi
-\frac12 g_{\mu\nu}(\nabla\phi)^2
-g_{\mu\nu}V(\phi),
\label{eq:scalar-stress}
\\
T_{\mu\nu}^{(F)}
&=
B(\phi)
\left(
F_{\mu\rho}F_\nu{}^\rho
-\frac14 g_{\mu\nu}F^2
\right),
\label{eq:gauge-stress}
\\
\nabla_\mu\!\left[B(\phi)F^{\mu\nu}\right]
&=0,
\label{eq:maxwell}
\\
\Box\phi
&=
V'(\phi)+\frac14 B'(\phi)F^2.
\label{eq:scalar-general}
\end{align}

We work throughout in the magnetic sector, for which
\begin{equation}
F
=
Q_m\sin\theta\,d\theta\wedge d\varphi,
\qquad
F_{\theta\varphi}=Q_m\sin\theta,
\label{eq:magnetic-ansatz}
\end{equation}
so that
\begin{equation}
F_{\mu\nu}F^{\mu\nu}
=
\frac{2Q_m^2}{r^4},
\label{eq:magnetic-invariant}
\end{equation}
and the scalar source generated by the gauge field is proportional to
$B'(\phi)$. For this ansatz, the Maxwell equation Eq.~\eqref{eq:maxwell} is
identically satisfied. In the absence of magnetic sources, the Bianchi
identity $dF=0$ implies
\[
\partial_vQ_m=\partial_rQ_m=0,
\]
so that $Q_m$ is an exact constant, rather than merely an
adiabatically conserved quantity. The electrically dual description
may be formulated in terms of the dual field strength and the
reciprocal gauge kinetic function $B(\phi)^{-1}$, but the
corresponding electric and dyonic systems are not analyzed here. No
use of this duality is made in the derivation below.

%%%%%%%%%%%%%%%%%%%%%%%%%%%%%%%%%%%%%%%%%%%%%%%%%%%%%%%%%%%%%

\subsection{Exact near-zone reduction}
\label{subsec:nearzone}

We work in the advanced Eddington--Finkelstein gauge
\begin{equation}
ds^2
=
-e^{2\delta(v,r)}\mathcal F(v,r)\,dv^2
+2e^{\delta(v,r)}\,dv\,dr
+r^2d\Omega_2^2,
\qquad
\mathcal F(v,r)
=
1-\frac{2m(v,r)}{\Mpl^2r}.
\label{eq:metric}
\end{equation}
Here $r$ is the exact areal radius, $v$ is an advanced-time
coordinate, $\delta(v,r)$ is a dimensionless lapse-type field, and
$m(v,r)$ is the Misner--Sharp quasi-local mass
function~\cite{Misner:1964je}. At this stage, both $\delta$ and $m$
are exact dynamical fields and are not assumed to depend on a finite
set of collective coordinates.

The chart, Eq.~\eqref{eq:metric}, is horizon-penetrating and remains regular
at a smooth future horizon provided $\delta$, $m$, and the matter
fields remain regular there. It retains the residual freedom
\begin{equation}
v\longrightarrow \widetilde v(v),
\end{equation}
under which
\begin{equation}
e^{\widetilde\delta(\widetilde v,r)}
=
e^{\delta(v,r)}
\frac{dv}{d\widetilde v}.
\label{eq:v-reparam}
\end{equation}
For the finite near-zone problem, we fix this freedom by imposing
\begin{equation}
\delta(v,r_{\rm m})=0
\label{eq:matching-normalization}
\end{equation}
at a fixed areal matching radius $r_{\rm m}$, taken to be independent
of $v$. Thus the normalization of the advanced-time coordinate is
fixed at the matching surface. We use this matching normalization
throughout Sections~\ref{sec:reduction}--\ref{sec:transition_companion}.

\begin{remark}[Normalization of static reference solutions]
Static reference solutions, denoted by a subscript $0$, are often
presented in the asymptotic normalization $\delta_0(\infty)=0$.
Passing from that normalization to
$\delta_0(r_{\rm m})=0$ requires a constant reparametrization of the
advanced time. Under such a transformation, the reduced radial
operator is multiplied by a nonzero constant, so its homogeneous
kernel is unchanged. Quantities that depend on the normalization of
the operator, its leading coefficient, or the time derivative must,
however, be converted explicitly. This conversion is required
whenever a closed-form static expression obtained in the asymptotic
normalization is used in the finite-$r_{\rm m}$ construction.
\end{remark}

The future outer marginal horizon is the marginally trapped tube
whose areal radius $r_{\rm H}(v)$ is defined by
\begin{equation}
\mathcal F\bigl(v,r_{\rm H}(v)\bigr)=0.
\label{eq:horizon-definition}
\end{equation}
Throughout this paper, the \emph{near zone} is the finite radial
region
\begin{equation}
r_{\rm H}(v)\leq r\leq r_{\rm m},
\qquad
r_{\rm H}\lesssim r_{\rm m}\ll L_{\rm cos},
\label{eq:near-zone-domain}
\end{equation}
where, in units with $c=1$, $L_{\rm cos}$ denotes the shortest
characteristic length scale associated with the evolution of the
exterior cosmological solution.

The condition $r_{\rm m}\ll L_{\rm cos}$ does not require the modulus
to undergo only a small total evolution over cosmological times. It
requires only that its variation across the finite near zone be small.
In particular, over a near-zone light-crossing time of order
$r_{\rm m}$,
\begin{equation}
\Delta\Phi
\sim
\dot\Phi\,r_{\rm m},
\qquad
\frac{|\Delta\Phi|}{\Mpl}\ll1.
\label{eq:near-zone-modulus-variation}
\end{equation}
The slowly varying exterior can therefore be represented locally, at
each advanced time, by a single prescribed modulus $\Phi(v)$ together
with its slowly varying derivatives.

Within the region defined in Eq.~\eqref{eq:near-zone-domain}, the black hole geometry
and all radial gradients are treated exactly. Cosmological evolution
enters only through slowly varying data supplied at $r=r_{\rm m}$.
The near-zone domain does not include spatial infinity. Consequently,
no asymptotic mass, such as the Arnowitt--Deser--Misner
mass~\cite{Arnowitt:1960es}, is defined as part of the near-zone
boundary-value problem. A global static solution may nevertheless be
used to generate the reference family whose profiles are subsequently
restricted to the finite interval $[r_{\rm H},r_{\rm m}]$.

Substituting Eq.~\eqref{eq:metric} and the magnetic
ansatz, Eq.~\eqref{eq:magnetic-ansatz}, into the field equations gives
\begin{equation}
\left\{
\begin{aligned}
\delta_{,r}
&=
\frac{r}{2\Mpl^2}\phi_{,r}^2,
\\[2mm]
m_{,r}
&=
\frac{r^2}{4}\mathcal F\,\phi_{,r}^2
+\frac{r^2}{2}V(\phi)
+\frac{B(\phi)Q_m^2}{4r^2},
\\[2mm]
m_{,v}
&=
\frac{r^2}{2}
\left[
e^{-\delta}\phi_{,v}^2
+\mathcal F\,\phi_{,v}\phi_{,r}
\right],
\\[2mm]
0
&=
\dd_v\!\left(r^2\phi_{,r}\right)
+\dd_r\!\left[
r^2\phi_{,v}
+e^\delta r^2\mathcal F\,\phi_{,r}
\right]
-e^\delta r^2\mathcal S(\phi,r),
\\[2mm]
\mathcal F
&=
1-\frac{2m}{\Mpl^2r},
\end{aligned}
\right.
\label{eq:system}
\end{equation}
where
\begin{equation}
\mathcal S(\phi,r)
\equiv
V'(\phi)+\frac{Q_m^2}{2r^4}B'(\phi).
\label{eq:scalar-source}
\end{equation}
Eq.~\eqref{eq:system} is exact. No adiabatic approximation and no
instantaneous-static ansatz have been used in its derivation. The
$\theta\theta$ Einstein equation is not independent: once the
equations displayed in Eq.~\eqref{eq:system} and the scalar equation
are satisfied, it follows from the contracted Bianchi identity and
serves as a consistency check.

A slowly varying cosmological scalar carries kinetic stress-energy of
order
\[
T^{(\phi)}_{vv}\sim\dot\Phi^2,
\]
while a nonzero potential contributes
\[
T^{(V)}_{\mu\nu}=-V(\Phi)g_{\mu\nu}.
\]
The exterior geometry is therefore not assumed to be exactly
asymptotically flat. At $r=r_{\rm m}$, the near-zone fields are
matched to a local expansion of the slowly varying exterior solution.
The prescribed quantity $\Phi(v)$ labels that exterior solution; it
need not coincide exactly with $\phi(v,r_{\rm m})$, since their
relation depends on the matching functional.

The characteristic exterior scale may be estimated as
\begin{equation}
L_{\rm cos}
\sim
\min\!\left(
T_{\dot\Phi},
L_V,
L_H
\right),
\qquad
T_{\dot\Phi}
\equiv
\frac{\Mpl}{|\dot\Phi|},
\qquad
L_V
\equiv
|V''(\Phi)|^{-1/2},
\qquad
L_H
\equiv
H^{-1},
\label{eq:Lcos}
\end{equation}
where $T_{\dot\Phi}$ is the canonical field-variation timescale,
$L_V$ is the local length scale associated with the curvature of the
potential, and $L_H$ is the Hubble radius of the exterior
cosmology.\footnote{One may also introduce the fractional variation
timescale $|\Phi/\dot\Phi|$, but it should not be identified with
$T_{\dot\Phi}$. They are parametrically comparable only when
$|\Phi|\sim\Mpl$, which is not assumed here.} Eq.~\eqref{eq:Lcos}
is an order-of-magnitude characterization of the shortest relevant
exterior scale, not an additional field equation.

The restriction to finite $r_{\rm m}$, rather than a global limit
$r\rightarrow\infty$, is part of the definition of the construction.
Any formal limit $r_{\rm m}\rightarrow\infty$ must therefore be
understood as a limit of expressions obtained from the finite-radius
problem, and its physical validity must be assessed separately.

We assign one power of adiabatic order to each derivative with
respect to the prescribed slow time dependence, while radial
derivatives are kept exact:
\begin{equation}
\dd_v={\mathcal O}_{\rm ad}(\epsilon),
\qquad
\dd_r={\mathcal O}_{\rm ad}(1).
\label{eq:adiabatic-counting-derivatives}
\end{equation}
Several dimensionless adiabatic parameters may be formed from the
ratio of the horizon scale to the exterior timescales in
Eq.~\eqref{eq:Lcos}. For example,
\begin{equation}
\epsilon_\Phi
\equiv
\frac{r_{\rm H}|\dot\Phi|}{\Mpl}
=
\frac{r_{\rm H}}{T_{\dot\Phi}}.
\label{eq:epsilon-Phi}
\end{equation}
These parameters need not be equal. In what follows, $\epsilon$
denotes the largest relevant adiabatic ratio. The notation
${\mathcal O}_{\rm ad}(\epsilon^n)$ refers to adiabatic order and
should not, by itself, be interpreted as a dimensionless estimate of
the quantity to which it is attached.

The exact mass-balance equation in Eq.~\eqref{eq:system} exhibits a
structural distinction between the bulk of the near zone and the
horizon. Away from the horizon, the mixed term
\begin{equation}
 \mathcal F\,\phi_{,v}\phi_{,r}   
\end{equation}
is generically of first adiabatic order, whereas
\begin{equation}
 e^{-\delta}\phi_{,v}^2   
\end{equation}
starts at second order. At the horizon,
$\mathcal F|_{\rm H}=0$ eliminates the mixed term exactly, leaving
\begin{equation}
m_{,v}\big|_{\rm H}
=
\frac{r_{\rm H}^2}{2}
e^{-\delta_{\rm H}}
\phi_{,v}^2\big|_{\rm H}
=
{\mathcal O}_{\rm ad}(\epsilon^2).
\label{eq:horizon-local-flux-order}
\end{equation}
Correspondingly,
\begin{equation}
\frac{r_{\rm H}\dot\Phi}{\Mpl}
=
{\mathcal O}_{\rm ad}(\epsilon),
\label{eq:Phi-counting}
\end{equation}
whereas the exact horizon-growth law derived below implies
\begin{equation}
\dot r_{\rm H}
=
{\mathcal O}_{\rm ad}(\epsilon^2).
\label{eq:counting}
\end{equation}

This hierarchy assumes that the horizon remains uniformly
non-degenerate on the adiabatic scale. In terms of the
normalization-dependent instantaneous surface gravity introduced in
Eq.~\eqref{eq:surface_gravity}, this requires
\begin{equation}
\kappa_{\rm H}r_{\rm H}\gg\epsilon.
\label{eq:nonextremal-adiabatic-condition}
\end{equation}
The ultra-near-extremal regime, in which
$\kappa_{\rm H}r_{\rm H}$ becomes comparable to $\epsilon$, lies
outside the scope of the present expansion.

%%%%%%%%%%%%%%%%%%%%%%%%%%%%%%%%%%%%%%%%%%%%%%%%%%%%%%%%%%%%%

\subsection{Exact horizon identities}
\label{subsec:horizon}

The future outer marginal horizon defined by
Eq.~\eqref{eq:horizon-definition} is the marginally trapped tube
$r=r_{\rm H}(v)$. A subscript ${\rm H}$ denotes evaluation on this
tube. Differentiating
\begin{equation}
\mathcal F\bigl(v,r_{\rm H}(v)\bigr)=0    
\end{equation}
with respect to $v$ gives
\begin{equation}
\mathcal F_{,v}\big|_{\rm H}
+
\dot r_{\rm H}\,
\mathcal F_{,r}\big|_{\rm H}
=
0.
\label{eq:differentiated-horizon-condition}
\end{equation}
Using
\[
\mathcal F_{,v}
=
-\frac{2m_{,v}}{\Mpl^2r}
\]
and the horizon value of the exact mass-balance equation,
Eq.~\eqref{eq:horizon-local-flux-order}, one obtains
\begin{equation}
\dot r_{\rm H}
=
\frac{
r_{\rm H}e^{-\delta_{\rm H}}
\phi_{,v}^2|_{\rm H}
}{
\Mpl^2\mathcal F_{,r}|_{\rm H}
}.
\label{eq:113}
\end{equation}
This identity is exact wherever
$\mathcal F_{,r}|_{\rm H}\neq0$.

With the advanced-time normalization fixed by
Eq.~\eqref{eq:matching-normalization}, we define the instantaneous
surface gravity parameter
\begin{equation}
\kappa_{\rm H}
\equiv
\frac12
e^{\delta_{\rm H}}
\mathcal F_{,r}\big|_{\rm H}.
\label{eq:surface_gravity}
\end{equation}
For a static solution, this reduces to the usual Killing surface
gravity in the same normalization of the time coordinate. For a
dynamical marginal horizon, Eq.~\eqref{eq:surface_gravity} should be
understood as the corresponding instantaneous, normalization-dependent
generalization; no claim of a unique notion of dynamical surface
gravity is required in the present analysis.

For a non-degenerate future outer marginal horizon,
\begin{equation}
\mathcal F_{,r}|_{\rm H}>0,
\qquad
\kappa_{\rm H}>0.
\label{eq:outer-horizon-condition}
\end{equation}
Eq.~\eqref{eq:113} then implies directly
\begin{equation}
\dot r_{\rm H}\geq0.
\label{eq:horizon-monotonicity}
\end{equation}
Thus, the marginal-horizon radius is non-decreasing, with equality
only when $\phi_{,v}|_{\rm H}=0$. This monotonicity follows from the
exact field equations and the future-outer condition; it does not
require the adiabatic approximation.

The inverse instantaneous surface gravity scale provides the natural
local black hole timescale associated with a uniformly
non-degenerate horizon,
\begin{equation}
\tau_{\rm BH}\sim\kappa_{\rm H}^{-1},
\label{eq:BH-timescale}
\end{equation}
while the slow exterior forcing varies on a characteristic timescale
\begin{equation}
T_{\rm ad}\sim\frac{r_{\rm H}}{\epsilon}.
\label{eq:adiabatic-timescale}
\end{equation}
The adiabatic regime requires a separation of scales,
\begin{equation}
\tau_{\rm BH}\ll T_{\rm ad},
\label{eq:timescale-separation}
\end{equation}
or equivalently
\begin{equation}
\kappa_{\rm H}r_{\rm H}\gg\epsilon.
\label{eq:timescale-separation-kappa}
\end{equation}
This condition expresses the assumption that the local black hole
geometry can adjust quasistatically to the slowly varying exterior.
It is an adiabatic validity condition, not an ingredient in the exact
monotonicity statement, Eq.~\eqref{eq:horizon-monotonicity}.

The defining equation of the marginal horizon gives
\begin{equation}
2m\bigl(v,r_{\rm H}(v)\bigr)
=
\Mpl^2r_{\rm H}(v)
\label{eq:horizon-mass-identity-pre}
\end{equation}
at every advanced time. Defining the exact horizon mass parameter by
\begin{equation}
M_{\rm H}^{\rm exact}(v)
\equiv
m\bigl(v,r_{\rm H}^{\rm exact}(v)\bigr),
\label{eq:horizon-mass-definition}
\end{equation}
where $r_{\rm H}^{\rm exact}(v)$ is the marginal-horizon radius defined
by Eq.~\eqref{eq:horizon-definition}, we obtain the exact kinematic
identity
\begin{equation}
M_{\rm H}^{\rm exact}(v)
=
\frac{\Mpl^2}{2}\,
r_{\rm H}^{\rm exact}(v).
\label{eq:MH-def}
\end{equation}
This relation holds for any spherically symmetric configuration
admitting a marginal horizon of the form defined above. It does not
rely on staticity and, apart from the definition of $\mathcal F$, does
not require the remaining field equations.

We refer to $M_{\rm H}^{\rm exact}$ as the \emph{exact horizon mass
parameter}, reserving the unqualified term ``mass'' for the quasi-local
function $m(v,r)$. This terminology also distinguishes it from the
horizon-mass parameter that will later be used to label an
instantaneous static representative. The latter is a collective
coordinate on the static solution manifold and need not coincide
pointwise with $M_{\rm H}^{\rm exact}$ in the time-dependent problem.

Combining Eq.~\eqref{eq:MH-def} with
Eq.~\eqref{eq:counting} gives
\begin{equation}
\dot M_{\rm H}^{\rm exact}
=
\frac{\Mpl^2}{2}
\dot r_{\rm H}^{\rm exact}
=
{\mathcal O}_{\rm ad}
\left(
\epsilon^2
\frac{M_{\rm H}^{\rm exact}}
{r_{\rm H}^{\rm exact}}
\right).
\label{eq:MH-counting}
\end{equation}
Thus, the exact horizon mass is protected against first-order evolution
by the quadratic character of the horizon flux. This property will
motivate, but should not be confused with, the choice of a slowly
varying horizon-mass coordinate on the static solution manifold in
Sections~\ref{sec:operator_and_geometry} and~\ref{sec:tracking}. The
relation between that instantaneous collective coordinate and
$M_{\rm H}^{\rm exact}$ is analyzed only after the lag fields have
been introduced.

Regularity at a smooth non-degenerate future horizon in the advanced
chart requires
\begin{equation}
 \phi_{\rm H},
\qquad
\phi_{,v}|_{\rm H},
\qquad
\phi_{,r}|_{\rm H},
\qquad
\delta_{\rm H}   
\end{equation}
to remain finite. Neither $\phi_{,v}|_{\rm H}=0$ nor
$\phi_{,r}|_{\rm H}=0$ is imposed.

Evaluating the exact scalar equation in
Eq.~\eqref{eq:system} at $\mathcal F=0$ gives
\begin{equation}
\mathcal F_{,r}|_{\rm H}\,
\phi_{,r}|_{\rm H}
=
\mathcal S_{\rm H}
-2e^{-\delta_{\rm H}}\phi_{,vr}|_{\rm H}
-\frac{2e^{-\delta_{\rm H}}}{r_{\rm H}^{\rm exact}}
\phi_{,v}|_{\rm H},
\label{eq:131}
\end{equation}
where
\begin{equation}
\mathcal S_{\rm H}
\equiv
\mathcal S(\phi,r)
\big|_{r=r_{\rm H}^{\rm exact}},
\qquad
\phi_{,vr}
\equiv
\frac{\partial^2\phi}{\partial v\,\partial r}.
\end{equation}
In the static limit, Eq.~\eqref{eq:131} reduces to
\begin{equation}
\mathcal F_{,r}|_{\rm H}\,
\phi_{,r}|_{\rm H}
=
V'(\phi_{\rm H})
+
\frac{Q_m^2}
{2\bigl(r_{\rm H}^{\rm exact}\bigr)^4}
B'(\phi_{\rm H}).
\label{eq:static-horizon-relation}
\end{equation}
In the dynamical problem it receives genuine contributions involving
$\phi_{,v}$ and $\phi_{,vr}$. These terms are not reproduced by merely
promoting the parameters of a static solution to slowly varying
functions of $v$.

At the linearized level, Eq.~\eqref{eq:131} determines the relation
between the horizon value and the first radial derivative of a regular
scalar perturbation. It therefore supplies the local
horizon-regularity condition used in defining the domain of the
reduced radial operator in
Section~\ref{sec:operator_and_geometry}.

%%%%%%%%%%%%%%%%%%%%%%%%%%%%%%%%%%%%%%%%%%%%%%%%%%%%%%%%%%%%%%%%%%%%%%%%%%%%

%%%%%%%%%%%%%%%%%%%%%%%%%%%%%%%%%%%%%%%%%%%%%%%%%%%%%%%%%%%%%%%%%%%%%%%%%%%%

%%%%%%%%%%%%%%%%%%%%%%%%%%%%%%%%%%%%%%%%%%%%%%%%%%%%%%%%%%%%%%%%%%%%%%%%%%%%

\setcounter{section}{2}
\section{The reduced scalar operator and canonical tangent geometry}
\label{sec:operator_and_geometry}

This section develops the two ingredients required for the subsequent
adiabatic construction. We first reduce the coupled linearized
Einstein--Maxwell--dilaton equations to a single local radial equation
for the scalar perturbation. We then study the action of the resulting
operator on tangent vectors to the manifold of static solutions and
identify the unique tangent direction generated by variations of the
static modulus label at fixed static horizon-mass parameter.
Throughout this section, all horizon quantities refer to the static
background. For a solution labeled by $\lambda\in\mathcal S$, its
horizon radius and horizon-mass parameter are denoted by
\begin{equation}
r_{{\rm H},0}(\lambda),
\qquad
M_{{\rm H},0}(\lambda)
\equiv
m_0\bigl(r_{{\rm H},0}(\lambda);\lambda\bigr)
=
\frac{\Mpl^2}{2}r_{{\rm H},0}(\lambda).
\label{eq:static-horizon-notation}
\end{equation}
After parameter promotion in Section~\ref{sec:tracking}, these become
$r_{\rm H}^{\rm inst}(v)$ and $M_{\rm H}^{\rm inst}(v)$. They should
not be confused with the exact dynamical quantities
$r_{\rm H}^{\rm exact}(v)$ and $M_{\rm H}^{\rm exact}(v)$ introduced
in Section~\ref{sec:reduction}.
%%%%%%%%%%%%%%%%%%%%%%%%%%%%%%%%%%%%%%%%%%%%%%%%%%%%%%%%%%%%%
\subsection{Reduction to a single scalar equation}
\label{subsec:scalar_reduction}
The Einstein equations form a first-order constraint system for the
metric perturbations. Once the scalar perturbation and the residual
integration data are specified, the metric perturbations are fixed,
while the scalar equation supplies the remaining evolution equation.
We therefore eliminate the metric variables $(\xi,\mu)$ and derive a
closed equation for the scalar perturbation $\eta$.
Let $\lambda_\star\in\mathcal S$ be a fixed point on the static
solution manifold. We write
\begin{equation}
\phi(v,r)
=
\phi_0(r)+\eta(v,r),
\qquad
\delta(v,r)
=
\delta_0(r)+\xi(v,r),
\qquad
m(v,r)
=
m_0(r)+\mu(v,r),
\label{eq:linear_perturbations}
\end{equation}
and retain terms only to first order in $(\eta,\xi,\mu)$. The
corresponding perturbation of the metric function is
\begin{equation}
\Delta\mathcal F(v,r)
=
-\frac{2\mu(v,r)}{\Mpl^2r}.
\label{eq:delta_F_linear}
\end{equation}
Here and below,
\begin{equation}
V_0^{(n)}(r)
\equiv
V^{(n)}\!\bigl(\phi_0(r)\bigr),
\qquad
B_0^{(n)}(r)
\equiv
B^{(n)}\!\bigl(\phi_0(r)\bigr)
\end{equation}
denote derivatives with respect to the scalar argument, evaluated on
the static background.
All background quantities in this subsection are evaluated at the
fixed point $\lambda_\star$. In particular, the background horizon is
the fixed radius
\begin{equation}
r_{{\rm H},0}
\equiv
r_{{\rm H},0}(\lambda_\star),
\qquad
\mathcal F_0(r_{{\rm H},0})=0.
\label{eq:fixed-static-horizon}
\end{equation}
The promotion of the static parameters to slowly varying collective
coordinates is performed only after the frozen-background operator has
been constructed. The reduction is nevertheless local on
$\mathcal S$ and applies at every smooth, non-extremal point of the
static family.
%%%%%%%%%%%%%%%%%%%%%%%%%%%%%%%%%%%%%%%%%%%%%%%%%%%%%%%%%%%%%
\subsubsection{Elimination of \texorpdfstring{$\xi$}{xi}}
The linearized radial Einstein constraint is
\begin{equation}
\xi_{,r}
=
\frac{r}{\Mpl^2}\phi_0'\eta_{,r}.
\label{eq:part3_xi_constraint}
\end{equation}
Together with the linearized matching normalization
\begin{equation}
\xi(v,r_{\rm m})=0,
\end{equation}
this gives
\begin{equation}
\xi(v,r)
=
\frac{1}{\Mpl^2}
\int_{r_{\rm m}}^r
d\bar r\,
\bar r\,\phi_0'(\bar r)\,
\eta_{,\bar r}(v,\bar r).
\label{eq:xi_integral}
\end{equation}
Thus $\xi$ is uniquely determined by $\eta$. Although
Eq.~\eqref{eq:xi_integral} is nonlocal in the radial coordinate, this
apparent nonlocality cancels from the reduced scalar equation. The
integral itself therefore never needs to be evaluated.
%%%%%%%%%%%%%%%%%%%%%%%%%%%%%%%%%%%%%%%%%%%%%%%%%%%%%%%%%%%%%
\subsubsection{Compatibility and elimination of
\texorpdfstring{$\mu$}{mu}}
\label{sec:compatibility}
For convenience, define
\begin{equation}
f_\delta(r)
\equiv
\frac{r\phi_0'^2}{2\Mpl^2},
\qquad
f_V(r)
\equiv
\frac{r^2}{2}V_0'
+
\frac{Q_m^2}{4r^2}B_0',
\qquad
f_\phi(r)
\equiv
\frac{r^2}{2}\mathcal F_0\phi_0'.
\label{eq:abc_definitions}
\end{equation}
Linearizing the radial and temporal mass equations in
Eq.~\eqref{eq:system} about the fixed static background gives
\begin{align}
\mu_{,r}
&=
-f_\delta\mu
+
f_V\eta
+
f_\phi\eta_{,r},
\label{eq:mass_radial_abc}
\\
\mu_{,v}
&=
f_\phi\eta_{,v}.
\label{eq:mass_temporal_abc}
\end{align}
A complete derivation is given in Appendix~\ref{app:linearization}.
Before solving these equations, one must verify their mutual
compatibility. Differentiating Eq.~\eqref{eq:mass_radial_abc} with
respect to $v$ and Eq.~\eqref{eq:mass_temporal_abc} with respect to
$r$ gives
\begin{align}
\partial_v\mu_{,r}
&=
\bigl(f_V-f_\delta f_\phi\bigr)\eta_{,v}
+
f_\phi\eta_{,vr},
\\
\partial_r\mu_{,v}
&=
f_\phi'\eta_{,v}
+
f_\phi\eta_{,vr}.
\end{align}
Hence the integrability condition is
\begin{equation}
\bigl(
f_V-f_\delta f_\phi-f_\phi'
\bigr)\eta_{,v}
=
0.
\label{eq:mass_integrability_pre}
\end{equation}
The coefficient in Eq.~\eqref{eq:mass_integrability_pre} vanishes
identically by virtue of the static scalar equation,
\begin{equation}
\partial_r
\left(
e^{\delta_0}r^2\mathcal F_0\phi_0'
\right)
=
e^{\delta_0}r^2
\left(
V_0'
+
\frac{Q_m^2}{2r^4}B_0'
\right).
\label{eq:static_scalar_part3}
\end{equation}
Indeed, using
\begin{equation}
\delta_0'
=
\frac{r\phi_0'^2}{2\Mpl^2}
=
f_\delta
\end{equation}
and dividing Eq.~\eqref{eq:static_scalar_part3} by
$2e^{\delta_0}$ gives
\begin{equation}
f_\phi'
+
f_\delta f_\phi
=
f_V.
\label{eq:background_identity_abc}
\end{equation}
Therefore
\begin{equation}
\partial_v\mu_{,r}
-
\partial_r\mu_{,v}
\equiv
0
\label{eq:mass_constraints_compatible}
\end{equation}
on every static background satisfying the exact field equations.
Compatibility is thus not an additional assumption: it is the
linearized consequence of the static scalar equation.
Eq.~\eqref{eq:mass_temporal_abc} implies
\begin{equation}
\partial_v\bigl(\mu-f_\phi\eta\bigr)=0,
\end{equation}
and hence
\begin{equation}
\mu(v,r)
=
f_\phi(r)\eta(v,r)+f_\mu(r),
\label{eq:mu_eta_h}
\end{equation}
where $f_\mu$ is independent of $v$. Substituting this expression into
Eq.~\eqref{eq:mass_radial_abc} gives
\begin{align}
f_\phi'\eta
+
f_\phi\eta_{,r}
+
f_\mu'
={}&
-f_\delta f_\phi\eta
-f_\delta f_\mu
+f_V\eta
+f_\phi\eta_{,r}.
\end{align}
The terms proportional to $\eta_{,r}$ cancel. The terms proportional
to $\eta$ cancel by Eq.~\eqref{eq:background_identity_abc}, leaving
\begin{equation}
f_\mu'
=
-f_\delta f_\mu.
\label{eq:h_radial_equation}
\end{equation}
Since $f_\delta=\delta_0'$, this integrates to
\begin{equation}
f_\mu(r)
=
C_\mu e^{-\delta_0(r)},
\label{eq:h_solution}
\end{equation}
where $C_\mu$ is independent of both $v$ and $r$. Consequently,
\begin{equation}
\mu(v,r)
=
\frac{r^2}{2}
\mathcal F_0\phi_0'\eta(v,r)
+
C_\mu e^{-\delta_0(r)}.
\label{eq:mu_reduced}
\end{equation}
The entire residual freedom in the mass perturbation is therefore
encoded in a single constant $C_\mu$, rather than in an arbitrary
function of $v$ or $r$. At this frozen-background stage no boundary
condition or physical interpretation is assigned to $C_\mu$. It
should in particular not yet be identified with a displacement of the
exact dynamical horizon. Its geometric meaning for tangent
deformations of the static family is determined below. Its
time-dependent analogue, introduced only after parameter promotion,
will encode the relative displacement between the exact horizon and
the horizon of the chosen instantaneous static representative.

%%%%%%%%%%%%%%%%%%%%%%%%%%%%%%%%%%%%%%%%%%%%%%%%%%%%%%%%%%%%%

\subsubsection{Cancellation of the nonlocal metric perturbation}
\label{sec:cancellation}

After substituting Eq.~\eqref{eq:mu_reduced}, the radial scalar-flux
perturbation appearing in the linearized scalar equation becomes
\begin{align}
\mathcal F_0\eta_{,r}
+\mathcal F_0\phi_0'\xi
-\frac{2\phi_0'}{\Mpl^2r}\mu
={}&
\mathcal F_0\eta_{,r}
+\mathcal F_0\phi_0'\xi
-\frac{r\mathcal F_0\phi_0'^2}{\Mpl^2}\eta
\nonumber\\
&\hspace{8mm}
-\frac{2C_\mu}{\Mpl^2r}
e^{-\delta_0}\phi_0'.
\label{eq:scalar_flux_after_mu}
\end{align}

The terms involving an undifferentiated lapse perturbation combine as
\begin{align}
&
\partial_r
\left(
e^{\delta_0}r^2\mathcal F_0\phi_0'\xi
\right)
-
e^{\delta_0}r^2
\left(
V_0'
+\frac{Q_m^2}{2r^4}B_0'
\right)\xi
\nonumber\\
&\hspace{20mm}
=
e^{\delta_0}r^2\mathcal F_0\phi_0'\xi_{,r},
\label{eq:xi_cancellation_step}
\end{align}
where the static scalar equation
\eqref{eq:static_scalar_part3} has been used. Thus all terms
proportional to an undifferentiated $\xi$ cancel identically, without
requiring the explicit integral representation
\eqref{eq:xi_integral}.

Using the radial constraint
\eqref{eq:part3_xi_constraint}, the remaining contribution involving
$\xi$ becomes
\begin{equation}
e^{\delta_0}
\frac{r^3\mathcal F_0\phi_0'^2}{\Mpl^2}
\eta_{,r}.
\label{eq:xi_local_term}
\end{equation}
On the other hand, differentiating the term proportional to $\eta$
generated by the substitution of $\mu$ gives
\begin{align}
-\partial_r
\left(
e^{\delta_0}
\frac{r^3\mathcal F_0\phi_0'^2}{\Mpl^2}
\eta
\right)
={}&
-e^{\delta_0}
\frac{r^3\mathcal F_0\phi_0'^2}{\Mpl^2}
\eta_{,r}
\nonumber\\
&-
\left[
e^{\delta_0}
\frac{r^3\mathcal F_0\phi_0'^2}{\Mpl^2}
\right]'
\eta.
\label{eq:eta_backreaction_derivative}
\end{align}
The first term on the right-hand side cancels
Eq.~\eqref{eq:xi_local_term} exactly. The apparent radial
nonlocality introduced by solving the lapse constraint therefore
disappears from the reduced scalar equation. Only a local
undifferentiated contribution to the effective radial potential
remains.

%%%%%%%%%%%%%%%%%%%%%%%%%%%%%%%%%%%%%%%%%%%%%%%%%%%%%%%%%%%%%

\subsubsection{The reduced scalar equation}

After the cancellations described above, the coupled linearized system
reduces to
\begin{align}
0
={}&
2r^2\eta_{,vr}
+2r\eta_{,v}
+\partial_r
\left(
e^{\delta_0}r^2\mathcal F_0\eta_{,r}
\right)
\nonumber\\
&-
\left\{
e^{\delta_0}r^2
\left[
V_0''
+\frac{Q_m^2}{2r^4}B_0''
\right]
+
\left[
e^{\delta_0}
\frac{r^3\mathcal F_0\phi_0'^2}{\Mpl^2}
\right]'
\right\}\eta
\nonumber\\
&-
\frac{2C_\mu}{\Mpl^2}
\bigl(r\phi_0'\bigr)'.
\label{eq:reduced_scalar_with_Cmu}
\end{align}

We define the effective radial potential by
\begin{equation}
\mathcal V_{\rm eff}(r)
\equiv
e^{\delta_0}r^2
\left[
V_0''
+\frac{Q_m^2}{2r^4}B_0''
\right]
+
\left[
e^{\delta_0}
\frac{r^3\mathcal F_0\phi_0'^2}{\Mpl^2}
\right]'
\label{eq:effective_radial_potential}
\end{equation}
and introduce the frozen-background radial operator in Sturm--Liouville
differential form,
\begin{equation}
L\eta
\equiv
\partial_r
\bigl(
p(r)\eta_{,r}
\bigr)
-
\mathcal V_{\rm eff}(r)\eta,
\qquad
p(r)
\equiv
e^{\delta_0(r)}r^2\mathcal F_0(r).
\label{eq:reduced_operator_definition}
\end{equation}
The reduced equation then takes the compact form
\begin{equation}
2r^2\eta_{,vr}
+
2r\eta_{,v}
+
L\eta
=
\frac{2C_\mu}{\Mpl^2}
\bigl(r\phi_0'\bigr)'.
\label{eq:reduced_scalar_compact}
\end{equation}

The reduction relies on two consequences of the same static background
equation. First, the radial and temporal mass constraints are mutually
compatible because their integrability condition is precisely
Eq.~\eqref{eq:background_identity_abc}, which follows from the static
scalar equation. Second, the same static equation cancels the
undifferentiated lapse perturbation and removes the apparent radial
nonlocality introduced by reconstructing $\xi$ from its constraint.

The metric sector therefore leaves only the residual constant
$C_\mu$. On the fixed static background, this constant multiplies the
fixed radial source profile
\begin{equation}
\frac{2}{\Mpl^2}
\bigl(r\phi_0'\bigr)'.
\end{equation}
The complete coupled linearized system is consequently reduced to the
single local scalar equation
\eqref{eq:reduced_scalar_compact}. At this stage, $C_\mu$ remains an
unspecified integration constant of the frozen-background problem; its
geometric interpretation for tangent perturbations is developed in the
following subsection.

%%%%%%%%%%%%%%%%%%%%%%%%%%%%%%%%%%%%%%%%%%%%%%%%%%%%%%%%%%%%%

\subsection{Canonical tangent theorem}
\label{subsec:canonical-tangent}

We now study the action of the reduced operator on tangent vectors to
the manifold of static solutions.

Let $\mathcal S$ be a smooth local two-dimensional manifold of
non-extremal static solutions at fixed magnetic charge $Q_m$, fixed
theory data $(V,B)$, and fixed matching normalization
\begin{equation}
\delta_0(r_{\rm m};\lambda)=0
\end{equation}
for every $\lambda\in\mathcal S$. We assume that the profiles
associated with neighboring points of $\mathcal S$ are defined on a
common radial neighborhood of the reference horizon, so that
parameter derivatives at fixed areal radius are well defined. The
tangent space $T_\lambda\mathcal S$ consists of infinitesimal
variations of the complete static field configuration at
$\lambda\in\mathcal S$.

Let $\lambda^A$ be arbitrary local coordinates on $\mathcal S$.
Differentiating the static fields at fixed areal radius defines
\begin{equation}
\varphi_A(r)
\equiv
\partial_A\phi_0(r;\lambda),
\qquad
\Delta_A(r)
\equiv
\partial_A\delta_0(r;\lambda),
\qquad
m_A(r)
\equiv
\partial_A m_0(r;\lambda).
\label{eq:def_varphiA}
\end{equation}
Since these fields are obtained by differentiating exact static
solutions, they satisfy the frozen-background linearized equations.
In particular,
\begin{equation}
L\varphi_A
=
\frac{2C_{\mu,A}}{\Mpl^2}
\bigl(r\phi_0'\bigr)',
\label{eq:def_LvarphiA}
\end{equation}
where $C_{\mu,A}$ is the residual mass-integration constant associated
with the tangent direction $\partial_A$.

The tangent version of Eq.~\eqref{eq:mu_reduced} is
\begin{equation}
\partial_A m_0(r)
=
f_\phi(r)\varphi_A(r)
+
C_{\mu,A}e^{-\delta_0(r)},
\label{eq:tangent_mass_decomposition}
\end{equation}
and hence
\begin{equation}
C_{\mu,A}
=
e^{\delta_0(r)}
\left[
\partial_A m_0(r)
-
f_\phi(r)\varphi_A(r)
\right].
\label{eq:Cmu_intrinsic_bulk}
\end{equation}
The right-hand side is independent of $r$ by the linearized
constraints and depends linearly on the tangent vector.

Under a change of coordinates
$\lambda^A\rightarrow\widetilde\lambda^A$ on $\mathcal S$,
\begin{equation}
\widetilde\varphi_A
=
\frac{\partial\lambda^B}
{\partial\widetilde\lambda^A}
\varphi_B,
\qquad
C_{\mu,\widetilde A}
=
\frac{\partial\lambda^B}
{\partial\widetilde\lambda^A}
C_{\mu,B}.
\label{eq:Cmu-transformation}
\end{equation}
Thus, once the matching normalization of the advanced-time coordinate
has been fixed, the coefficients $C_{\mu,A}$ are the components of a
covector on the static solution manifold,
\begin{equation}
\mathfrak C_\mu
\in
T_\lambda^\ast\mathcal S,
\qquad
\mathfrak C_\mu(\partial_A)
=
C_{\mu,A}.
\label{eq:Cmu-covector}
\end{equation}
In particular,
\begin{equation}
\ker\mathfrak C_\mu
\subset
T_\lambda\mathcal S
\end{equation}
is independent of the choice of collective coordinates. A constant
renormalization of the advanced-time coordinate rescales
$\mathfrak C_\mu$ by a nonzero factor and therefore leaves its kernel
unchanged.

For each $\lambda\in\mathcal S$, let $r_{{\rm H},0}(\lambda)$ denote
the static horizon radius, defined by
\begin{equation}
\mathcal F_0\bigl(r_{{\rm H},0}(\lambda);\lambda\bigr)=0.
\label{eq:static-horizon-definition-section3}
\end{equation}
At the static horizon,
\begin{equation}
f_\phi(r_{{\rm H},0})
=
\frac{r_{{\rm H},0}^2}{2}
\mathcal F_0(r_{{\rm H},0})
\phi_0'(r_{{\rm H},0})
=
0.
\end{equation}
Eq.~\eqref{eq:Cmu_intrinsic_bulk} therefore gives
\begin{equation}
C_{\mu,A}
=
e^{\delta_{0{\rm H}}}
\partial_A m_0(r_{{\rm H},0}),
\label{eq:Cmu-fixed-r-horizon}
\end{equation}
where
\begin{equation}
\delta_{0{\rm H}}
\equiv
\delta_0(r_{{\rm H},0};\lambda),
\end{equation}
and $\partial_A m_0(r_{{\rm H},0})$ means that the parameter derivative
is first taken at fixed $r$ and only then evaluated at
$r=r_{{\rm H},0}(\lambda)$.

The static horizon-mass parameter is
\begin{equation}
M_{{\rm H},0}(\lambda)
\equiv
m_0\bigl(r_{{\rm H},0}(\lambda);\lambda\bigr)
=
\frac{\Mpl^2}{2}r_{{\rm H},0}(\lambda).
\label{eq:static-MH}
\end{equation}
The chain rule gives
\begin{equation}
\partial_A M_{{\rm H},0}
=
\partial_A m_0(r_{{\rm H},0})
+
m_0'(r_{{\rm H},0})
\partial_A r_{{\rm H},0}.
\label{eq:static-MH-chain-rule}
\end{equation}
Combining this relation with
Eq.~\eqref{eq:Cmu-fixed-r-horizon} yields
\begin{equation}
C_{\mu,A}
=
e^{\delta_{0{\rm H}}}
\left[
\partial_A M_{{\rm H},0}
-
m_0'(r_{{\rm H},0})
\partial_A r_{{\rm H},0}
\right].
\label{eq:Cmu_horizon_formula}
\end{equation}
This horizon formula is the geometric core of the result below.

We now use the static modulus label $\Phi$ and the static
horizon-mass parameter $M_{{\rm H},0}$ as local coordinates on
$\mathcal S$. The coordinate $\Phi$ labels the static reference
family whose later dynamical promotion is matched to the prescribed
exterior modulus. We assume that
\begin{equation}
\bigl(\Phi,M_{{\rm H},0}\bigr)
\end{equation}
form regular local coordinates on $\mathcal S$, so that the
corresponding coordinate tangent vectors are nonzero and linearly
independent.

\begin{theorem}[Canonical Tangent Theorem]
\label{thm:canonical_tangent}
Let
$\bigl(\Phi,M_{{\rm H},0}\bigr)$
be regular local coordinates on the two-dimensional static solution
manifold $\mathcal S$. At every non-extremal future outer static
horizon,
\begin{equation}
\ker\mathfrak C_\mu
=
\operatorname{span}
\left\{
\left.
\partial_\Phi
\right|_{M_{{\rm H},0}}
\right\}.
\label{eq:canonical-kernel}
\end{equation}
Equivalently,
\begin{equation}
C_{\mu,\Phi}=0,
\qquad
C_{\mu,M_{{\rm H},0}}
=
e^{\delta_{0{\rm H}}}
r_{{\rm H},0}
\mathcal F_0'(r_{{\rm H},0})
\neq0.
\label{eq:Cmu-adapted-components}
\end{equation}
Consequently, the scalar component of the canonical tangent,
\begin{equation}
\varphi_\Phi
\equiv
\left.
\partial_\Phi\phi_0
\right|_{M_{{\rm H},0}},
\label{eq:canonical-scalar-tangent}
\end{equation}
satisfies
\begin{equation}
L\varphi_\Phi=0.
\label{eq:canonical-tangent-zero-mode}
\end{equation}
\end{theorem}

\begin{proof}
Along the direction
$\left.\partial_\Phi\right|_{M_{{\rm H},0}}$, one has
\begin{equation}
\left.
\partial_\Phi M_{{\rm H},0}
\right|_{M_{{\rm H},0}}
=
0.
\end{equation}
Moreover,
\begin{equation}
r_{{\rm H},0}
=
\frac{2M_{{\rm H},0}}{\Mpl^2}
\end{equation}
depends only on $M_{{\rm H},0}$, and therefore
\begin{equation}
\left.
\partial_\Phi r_{{\rm H},0}
\right|_{M_{{\rm H},0}}
=
0.
\end{equation}
Eq.~\eqref{eq:Cmu_horizon_formula} then gives
\begin{equation}
C_{\mu,\Phi}=0.
\end{equation}

Along the second adapted coordinate direction,
\begin{equation}
\left.
\partial_{M_{{\rm H},0}}M_{{\rm H},0}
\right|_\Phi
=
1,
\qquad
\left.
\partial_{M_{{\rm H},0}}r_{{\rm H},0}
\right|_\Phi
=
\frac{2}{\Mpl^2}.
\end{equation}
Hence
\begin{equation}
C_{\mu,M_{{\rm H},0}}
=
e^{\delta_{0{\rm H}}}
\left[
1-
\frac{2m_0'(r_{{\rm H},0})}{\Mpl^2}
\right].
\label{eq:Cmu-MH-intermediate}
\end{equation}
Differentiating
\begin{equation}
\mathcal F_0(r)
=
1-
\frac{2m_0(r)}{\Mpl^2r}
\end{equation}
and using
\begin{equation}
2m_0(r_{{\rm H},0})
=
\Mpl^2r_{{\rm H},0}
\end{equation}
gives
\begin{equation}
r_{{\rm H},0}
\mathcal F_0'(r_{{\rm H},0})
=
1-
\frac{2m_0'(r_{{\rm H},0})}{\Mpl^2}.
\end{equation}
Therefore
\begin{equation}
C_{\mu,M_{{\rm H},0}}
=
e^{\delta_{0{\rm H}}}
r_{{\rm H},0}
\mathcal F_0'(r_{{\rm H},0}).
\end{equation}
For a non-extremal future outer static horizon,
\begin{equation}
\mathcal F_0'(r_{{\rm H},0})>0,
\end{equation}
and hence
\begin{equation}
C_{\mu,M_{{\rm H},0}}\neq0.
\end{equation}
Since $T_\lambda\mathcal S$ is two-dimensional,
$\mathfrak C_\mu$ has rank one and its kernel is exactly the line
generated by
$\left.\partial_\Phi\right|_{M_{{\rm H},0}}$.

Finally, Eq.~\eqref{eq:def_LvarphiA} with
$C_{\mu,\Phi}=0$ gives
\begin{equation}
L\varphi_\Phi=0.
\end{equation}
\end{proof}

The differentiability of $r_{{\rm H},0}(\lambda)$ used above follows
from the implicit function theorem applied to
\begin{equation}
\mathcal F_0(r_{{\rm H},0};\lambda)=0,
\end{equation}
since non-extremality implies
\begin{equation}
\mathcal F_0'(r_{{\rm H},0})\neq0.
\end{equation}

\begin{remark}[Geometric versus functional statements]
\label{rem:functional}
The theorem is a statement about the tangent space of the static
solution manifold. It identifies the unique tangent direction whose
residual mass-integration constant vanishes. It does not, by itself,
establish that the associated scalar profile belongs to the domain of
a completed radial boundary-value problem. Horizon regularity and the
outer matching functional are incorporated in the functional-analytic
construction developed later.
\end{remark}

\begin{remark}[Non-extremality]
\label{rem:nonextremal}
At an extremal static horizon,
\begin{equation}
\mathcal F_0'(r_{{\rm H},0})=0,
\end{equation}
the implicit-function argument may fail and
$C_{\mu,M_{{\rm H},0}}$ vanishes according to
Eq.~\eqref{eq:Cmu-adapted-components}. The rank-one conclusion for
$\mathfrak C_\mu$ therefore need not survive the extremal limit. The
theorem is restricted throughout to non-extremal future outer static
horizons. 
Whether an analogous Fredholm construction can be formulated directly on the
extremal throat, where the horizon becomes degenerate and the implicit
function argument used above no longer applies, remains an interesting open
problem beyond the scope of the present work.

\end{remark}

%%%%%%%%%%%%%%%%%%%%%%%%%%%%%%%%%%%%%%%%%%%%%%%%%%%%%%%%%%%%%%%%%%%%%%%%%%%%

%%%%%%%%%%%%%%%%%%%%%%%%%%%%%%%%%%%%%%%%%%%%%%%%%%%%%%%%%%%%%%%%%%%%%%%%%%%%

%%%%%%%%%%%%%%%%%%%%%%%%%%%%%%%%%%%%%%%%%%%%%%%%%%%%%%%%%%%%%%%%%%%%%%%%%%%%

\setcounter{section}{3}
\section{Adiabatic tracking and the lag equation}
\label{sec:tracking}

This section constructs the leading adiabatic deformation of the
static family. We first promote the collective coordinates and compute
the defect of the resulting instantaneous-tracking configuration. We
then introduce the lag field, analyze the tangential ambiguity in its
definition, and reduce the first-order system to a single unsliced
scalar equation. The local slice and the completed radial
boundary-value problem are constructed in Section~\ref{sec:slice}.

%%%%%%%%%%%%%%%%%%%%%%%%%%%%%%%%%%%%%%%%%%%%%%%%%%%%%%%%%%%%%

\subsection{Promotion of the static family}
\label{subsec:part5_promotion}

Let
\begin{equation}
\Psi_0(r;\lambda)
\equiv
\bigl(
\phi_0(r;\lambda),
\delta_0(r;\lambda),
m_0(r;\lambda)
\bigr),
\qquad
\lambda^A
=
\bigl(
\Phi,M_{{\rm H},0}
\bigr),
\label{eq:part5_static_family_collective}
\end{equation}
be the smooth two-parameter family of exact static magnetic Einstein--Maxwell--dilaton
solutions introduced in Section~\ref{sec:operator_and_geometry}. The theory data $(V,B)$ and the
magnetic charge $Q_m$ are held fixed throughout the family. The
areal radius $r$ is held fixed when differentiating with respect to
the collective coordinates, and we define
\begin{equation}
\varphi_A
\equiv
\partial_A\phi_0,
\qquad
\Delta_A
\equiv
\partial_A\delta_0,
\qquad
m_A
\equiv
\partial_A m_0.
\label{eq:part5_tangent_fields}
\end{equation}
These tangent fields satisfy the frozen-background linearized
equations derived in Section~\ref{sec:operator_and_geometry}.

The coordinates $\lambda^A$ label the static representative used in
the adiabatic decomposition. They are introduced independently of the
time-dependent equations: $\Psi_0(r;\lambda)$ is a time-independent
map from the static solution manifold $\mathcal S$ to radial field
profiles. Time dependence is introduced only by promoting the
collective coordinates to slowly varying functions of advanced time,
\begin{equation}
\lambda^A
=
\lambda^A(v).
\end{equation}

After promotion, we denote the static horizon radius and static
horizon-mass parameter of the instantaneous representative by
\begin{equation}
r_{\rm H}^{\rm inst}(v)
\equiv
r_{{\rm H},0}\bigl(\lambda(v)\bigr),
\qquad
M_{\rm H}^{\rm inst}(v)
\equiv
M_{{\rm H},0}\bigl(\lambda(v)\bigr)
=
\frac{\Mpl^2}{2}r_{\rm H}^{\rm inst}(v).
\label{eq:instantaneous-horizon-parameters}
\end{equation}
These quantities belong to the chosen instantaneous static
representative. They should not be confused with the exact dynamical
horizon radius and mass parameter,
\begin{equation}
r_{\rm H}^{\rm exact}(v),
\qquad
M_{\rm H}^{\rm exact}(v)
=
\frac{\Mpl^2}{2}r_{\rm H}^{\rm exact}(v),
\end{equation}
introduced in Section~\ref{sec:reduction}. The instantaneous and exact horizon
quantities may differ at first adiabatic order once the lag fields are
included.

The exterior solution supplies a prescribed slowly varying modulus,
which we denote by
\begin{equation}
\Phi_{\rm ext}(v),
\qquad
\frac{
r_{\rm H}^{\rm inst}
|\dot\Phi_{\rm ext}|
}{\Mpl}
=
{\mathcal O}_{\rm ad}(\epsilon).
\label{eq:external-modulus-counting}
\end{equation}
Replacing $r_{\rm H}^{\rm inst}$ in this estimate by
$r_{\rm H}^{\rm exact}$ would change it only beyond the order retained
here. The quantity $\Phi_{\rm ext}$ labels the exterior configuration
to which the near zone is matched; it need not coincide directly with
the scalar value at $r=r_{\rm m}$. Its use as a single modulus across the finite near zone relies on the overlap hierarchy introduced in Section~\ref{sec:reduction}, which ensures that its variation across the matching region is negligible.

The promoted modulus coordinate $\lambda^\Phi(v)$ will be identified
with $\Phi_{\rm ext}(v)$ up to the residual tangential freedom analyzed
in Section~\ref{subsec:part5_lag_slice}. Admissible changes of
representative are required to vary only on the prescribed slow
timescale, so their effect on the collective velocity begins beyond
the order retained here.

The second promoted coordinate is
$M_{\rm H}^{\rm inst}(v)$. It labels the horizon mass of the
instantaneous static representative and is not identified pointwise
with $M_{\rm H}^{\rm exact}(v)$. We choose it within the slowly varying
sector,
\begin{equation}
\dot M_{\rm H}^{\rm inst}
=
{\mathcal O}_{\rm ad}(\epsilon^2),
\qquad
\dot r_{\rm H}^{\rm inst}
=
{\mathcal O}_{\rm ad}(\epsilon^2).
\label{eq:part5_counting}
\end{equation}
This counting is motivated by the exact quadratic horizon-flux law
satisfied by $M_{\rm H}^{\rm exact}$, but it is a condition imposed on
the choice of instantaneous representative rather than a pointwise
identification of the two horizon masses.

We anticipate here a residual tangential freedom in the choice of
instantaneous representative. This freedom is defined precisely in
Section~\ref{subsec:part5_lag_slice}, once the lag field and the equivalence
between nearby representatives have been introduced. Denoting the components
of such a first-order representative shift by \(\alpha^A(v)\) (see
Eq.~\eqref{eq:part5_shifted_velocity}), we fix its static horizon-mass
component by the convention
\begin{equation}
\alpha^{M_{{\rm H},0}}(v)\equiv0.
\label{eq:no-MH-tangent-shift-promotion}
\end{equation}
This convention is imposed identically and does not follow from the
slow-variation condition alone. Its consistency, and the relation between
the exact and instantaneous horizons, are established after the lag field
has been introduced.

The ultra-near-extremal regime, in which
\begin{equation}
\kappa_{\rm H}^{\rm inst}
r_{\rm H}^{\rm inst}
\sim
\epsilon,
\end{equation}
lies outside the scope of the expansion. Here
$\kappa_{\rm H}^{\rm inst}$ denotes the surface gravity parameter of
the instantaneous static representative in the matching
normalization.

We now write
\begin{equation}
\lambda^A=\lambda^A(v),
\qquad
\dot\lambda^A=\frac{d\lambda^A}{dv},
\end{equation}
and define the \emph{instantaneous-tracking configuration} by
\begin{equation}
\Psi_{\rm inst}(v,r)
\equiv
\Psi_0\bigl(r;\lambda(v)\bigr).
\label{eq:part5_instantaneous_configuration}
\end{equation}
Its metric function is
\begin{equation}
\mathcal F_{\rm inst}(v,r)
=
1-
\frac{
2m_0(r;\lambda(v))
}{
\Mpl^2r
}
\equiv
\mathcal F_0(r;\lambda(v)).
\end{equation}
By construction,
\begin{equation}
\mathcal F_{\rm inst}
\bigl(
v,r_{\rm H}^{\rm inst}(v)
\bigr)
=
0.
\end{equation}
Eq.~\eqref{eq:part5_instantaneous_configuration} is only a
definition of the instantaneous static representative. It neither
asserts that the exact dynamical solution remains on $\mathcal S$ nor
that
\begin{equation}
r_{\rm H}^{\rm exact}
=
r_{\rm H}^{\rm inst}.
\end{equation}

At fixed areal radius,
\begin{align}
\partial_v\phi_{\rm inst}
&=
\dot\lambda^A\varphi_A,
\\
\partial_v\delta_{\rm inst}
&=
\dot\lambda^A\Delta_A,
\\
\partial_vm_{\rm inst}
&=
\dot\lambda^A m_A,
\end{align}
and
\begin{equation}
\partial_v\partial_r\phi_{\rm inst}
=
\dot\lambda^A\varphi_A'.
\end{equation}
No acceleration $\ddot\lambda^A$ appears at this stage because the
exact minimal system~\eqref{eq:system} contains only first derivatives
with respect to $v$.

Using Eq.~\eqref{eq:part5_counting}, together with the adiabatic
restriction on changes of representative introduced below, the
first-order chain rules reduce to
\begin{align}
\partial_v\phi_{\rm inst}
&=
\dot\Phi_{\rm ext}\,\varphi_\Phi
+
{\mathcal O}_{\rm ad}(\epsilon^2),
\\
\partial_v\delta_{\rm inst}
&=
\dot\Phi_{\rm ext}\,\Delta_\Phi
+
{\mathcal O}_{\rm ad}(\epsilon^2),
\\
\partial_vm_{\rm inst}
&=
\dot\Phi_{\rm ext}\,m_\Phi
+
{\mathcal O}_{\rm ad}(\epsilon^2),
\label{eq:part5_first_order_chain_rule}
\end{align}
where
\begin{equation}
\varphi_\Phi
=
\left.
\partial_\Phi\phi_0
\right|_{M_{{\rm H},0}},
\qquad
\Delta_\Phi
=
\left.
\partial_\Phi\delta_0
\right|_{M_{{\rm H},0}},
\qquad
m_\Phi
=
\left.
\partial_\Phi m_0
\right|_{M_{{\rm H},0}}.
\label{eq:pderivatives}
\end{equation}
In Eq.~\eqref{eq:part5_first_order_chain_rule}, the terms proportional
to $\dot M_{\rm H}^{\rm inst}$ are of second adiabatic order and have
therefore been absorbed into the remainder. Moreover,
\begin{equation}
\dot\lambda^\Phi
=
\dot\Phi_{\rm ext}
+
{\mathcal O}_{\rm ad}(\epsilon^2),
\end{equation}
because admissible shifts of the modulus representative alter its
velocity only beyond first order. The distinction between
$\lambda^\Phi$ and $\Phi_{\rm ext}$ is consequently irrelevant to the
leading forcing.

%%%%%%%%%%%%%%%%%%%%%%%%%%%%%%%%%%%%%%%%%%%%%%%%%%%%%%%%%%%%%

\subsection{Defect of instantaneous tracking}
\label{subsec:part5_defect}

For each equation in the exact near-zone system, we define the defect
as its left-hand side evaluated on $\Psi_{\rm inst}$. By convention,
the defect vanishes on every exact solution. Since the static
equations hold for every fixed value of $\lambda$, all purely radial
terms cancel pointwise after parameter promotion.

%%%%%%%%%%%%%%%%%%%%%%%%%%%%%%%%%%%%%%%%%%%%%%%%%%%%%%%%%%%%%

\subsubsection{Radial constraints}

Define
\begin{align}
\mathcal E_\delta
&\equiv
\delta_{,r}
-\frac{r}{2\Mpl^2}\phi_{,r}^2,
\\
\mathcal E_{m,r}
&\equiv
m_{,r}
-\frac{r^2}{4}\mathcal F\phi_{,r}^2
-\frac{r^2}{2}V(\phi)
-\frac{Q_m^2}{4r^2}B(\phi).
\end{align}
The static equations imply
\begin{equation}
\mathcal E_\delta[\Psi_{\rm inst}]
=
0,
\qquad
\mathcal E_{m,r}[\Psi_{\rm inst}]
=
0
\label{eq:promoted-radial-constraints}
\end{equation}
exactly, at every $v$. Parameter promotion therefore preserves the
two radial constraints; nonzero defects can arise only in equations
containing explicit $v$-derivatives.

%%%%%%%%%%%%%%%%%%%%%%%%%%%%%%%%%%%%%%%%%%%%%%%%%%%%%%%%%%%%%

\subsubsection{Mass-balance defect and the canonical covector}

The exact mass-balance equation is
\begin{equation}
\mathcal E_{m,v}
\equiv
m_{,v}
-
\frac{r^2}{2}
\left[
e^{-\delta}\phi_{,v}^2
+
\mathcal F\phi_{,v}\phi_{,r}
\right].
\end{equation}
Evaluating it on the promoted static family gives
\begin{align}
\mathcal E_{m,v}[\Psi_{\rm inst}]
={}&
\dot\lambda^A m_A
-
\frac{r^2}{2}
\mathcal F_0\phi_0'
\dot\lambda^A\varphi_A
\nonumber\\
&-
\frac{r^2}{2}
e^{-\delta_0}
\dot\lambda^A\dot\lambda^B
\varphi_A\varphi_B
\nonumber\\
={}&
\dot\lambda^A
\bigl(
m_A-f_\phi\varphi_A
\bigr)
-
\frac{r^2}{2}
e^{-\delta_0}
\dot\lambda^A\dot\lambda^B
\varphi_A\varphi_B,
\label{eq:part5_Emv_before_C}
\end{align}
where
\begin{equation}
f_\phi(r)
=
\frac{r^2}{2}
\mathcal F_0(r)\phi_0'(r).
\end{equation}
Using the tangent decomposition
\begin{equation}
m_A-f_\phi\varphi_A
=
C_{\mu,A}e^{-\delta_0},
\end{equation}
derived in Section~\ref{sec:operator_and_geometry}, we obtain
\begin{equation}
\mathcal E_{m,v}[\Psi_{\rm inst}]
=
e^{-\delta_0}
\left[
\dot\lambda^A C_{\mu,A}
-
\frac{r^2}{2}
\dot\lambda^A\dot\lambda^B
\varphi_A\varphi_B
\right].
\label{eq:part5_Emv_exact}
\end{equation}

The second term is quadratic in the collective velocity and starts at
${\mathcal O}_{\rm ad}(\epsilon^2)$. The first-order defect is
therefore
\begin{equation}
\mathcal E_{m,v}^{(1)}[\Psi_{\rm inst}]
=
e^{-\delta_0}
\mathfrak C_\mu(\dot\lambda)
=
e^{-\delta_0}
\dot\lambda^A C_{\mu,A}.
\label{eq:part5_Emv_linear_intrinsic}
\end{equation}
Thus, the covector $\mathfrak C_\mu$ constructed in Section~\ref{sec:operator_and_geometry} measures
the first-order failure of a chosen tangent motion through
$\mathcal S$ to satisfy the exact mass-balance equation.

In the adapted coordinates
$\bigl(\Phi,M_{{\rm H},0}\bigr)$,
the Canonical Tangent Theorem gives
\begin{equation}
C_{\mu,\Phi}=0,
\qquad
C_{\mu,M_{{\rm H},0}}
=
e^{\delta_{0{\rm H}}}
r_{{\rm H},0}\,
\mathcal F_0'(r_{{\rm H},0})
>0.
\label{eq:part5_C_components_chart}
\end{equation}
for a non-degenerate future outer horizon. In particular,
\begin{equation}
\ker\mathfrak C_\mu
=
\operatorname{span}
\left\{
\left.
\partial_\Phi
\right|_{M_{{\rm H},0}}
\right\}.
\label{eq:part5-canonical-kernel}
\end{equation}
Consequently,
\begin{equation}
\mathfrak C_\mu(\dot\lambda)
=
C_{\mu,\Phi}\dot\lambda^\Phi
+
C_{\mu,M_{{\rm H},0}}\, \dot M_{\rm H}^{\rm inst}
=
{\mathcal O}_{\rm ad}(\epsilon^2),
\label{eq:part5_C_velocity_counting}
\end{equation}
and hence
\begin{equation}
\mathcal E_{m,v}^{(1)}[\Psi_{\rm inst}]
=
0.
\label{eq:first-order-mass-defect-zero}
\end{equation}

This cancellation is not produced by an imposed projection.
Rather, it follows from two geometric properties of the chosen
instantaneous representative. First, the tangent direction
\begin{equation}
 \left.
\partial_\Phi
\right|_{M_{{\rm H},0}}   
\end{equation}
belongs to
$\ker\mathfrak C_\mu$.
Second, the promoted static horizon-mass parameter
$M_{\rm H}^{\rm inst}(v)$
is chosen to vary only at second adiabatic order,
\begin{equation}
\dot M_{\rm H}^{\rm inst}
=
{\mathcal O}_{\rm ad}(\epsilon^2).    
\end{equation}
These two properties imply
$\mathfrak C_\mu(\dot\lambda)
=
{\mathcal O}_{\rm ad}(\epsilon^2)$,
and hence the first-order mass-balance defect vanishes.

%%%%%%%%%%%%%%%%%%%%%%%%%%%%%%%%%%%%%%%%%%%%%%%%%%%%%%%%%%%%%

\subsubsection{Scalar-equation defect}

Write the exact scalar equation as
\begin{equation}
\mathcal E_\phi
\equiv
\partial_v\!\left(r^2\phi_{,r}\right)
+
\partial_r
\left[
r^2\phi_{,v}
+
e^\delta r^2\mathcal F\phi_{,r}
\right]
-
e^\delta r^2\mathcal S(\phi,r),
\label{eq:scalar-defect-definition}
\end{equation}
where
\begin{equation}
\mathcal S(\phi,r)
=
V'(\phi)
+
\frac{Q_m^2}{2r^4}B'(\phi).
\end{equation}
For every fixed $\lambda$, the static family satisfies
\begin{equation}
\partial_r
\left(
e^{\delta_0}r^2\mathcal F_0\phi_0'
\right)
=
e^{\delta_0}r^2\mathcal S(\phi_0,r).
\label{eq:promoted-static-scalar-equation}
\end{equation}
The radial contribution therefore continues to cancel pointwise after
parameter promotion, since the promoted configuration is, at every
instant, an exact member of the static family evaluated at the current
collective coordinates.
The remaining explicit time derivatives give
\begin{align}
\mathcal E_\phi[\Psi_{\rm inst}]
&=
\dot\lambda^A r^2\varphi_A'
+
\partial_r
\left(
r^2\dot\lambda^A\varphi_A
\right)
\nonumber\\
&=
2\dot\lambda^A
\left(
r^2\varphi_A'
+
r\varphi_A
\right).
\label{eq:part5_scalar_defect_boxed}
\end{align}
This expression is exact in the collective velocities. The scalar
equation contains time derivatives only linearly, so no quadratic
term in $\dot\lambda^A$ is generated by parameter promotion.

At first adiabatic order,
\begin{equation}
\mathcal E_\phi^{(1)}[\Psi_{\rm inst}]
=
2\dot\lambda^\Phi
\left(
r^2\varphi_\Phi'
+
r\varphi_\Phi
\right),
\label{eq:part5_scalar_defect_first_order}
\end{equation}
and therefore
\begin{equation}
\mathcal E_\phi^{(1)}[\Psi_{\rm inst}]
=
2\dot\Phi_{\rm ext}
\left(
r^2\varphi_\Phi'
+
r\varphi_\Phi
\right)
+
{\mathcal O}_{\rm ad}(\epsilon^2).
\label{eq:part5-scalar-defect-external}
\end{equation}
The first-order scalar defect is therefore determined by the
prescribed exterior rolling rate and by the sensitivity of the static
scalar profile to variations of the modulus at fixed static
horizon-mass parameter
$M_{{\rm H},0}$.

%%%%%%%%%%%%%%%%%%%%%%%%%%%%%%%%%%%%%%%%%%%%%%%%%%%%%%%%%%%%%

\subsection{Lag field and tangential ambiguity}
\label{subsec:part5_lag_slice}

The promoted configuration satisfies the radial constraints exactly
and the mass-balance equation through first adiabatic order, but in
general it does not satisfy the scalar equation at that order. The
exact dynamical solution must therefore depart from the static
solution manifold.

We write
\begin{equation}
\Psi(v,r)
=
\Psi_{\rm inst}(v,r)
+
\Psi_{\rm lag}(v,r)
+
{\mathcal O}_{\rm ad}(\epsilon^2),
\qquad
\Psi_{\rm lag}
=
{\mathcal O}_{\rm ad}(\epsilon).
\label{eq:part5_lag_definition}
\end{equation}
At this stage, $\Psi_{\rm lag}$ is defined only as the first-order
difference between the exact dynamical solution and the chosen
instantaneous static representative. No equation, boundary condition,
or transversality condition has yet been imposed on it.

This decomposition is not unique. Consider a nearby representative curve
\begin{equation}
\widetilde\lambda^A(v)
=
\lambda^A(v)+\alpha^A(v).
\label{eq:part5_collective_shift}
\end{equation}
An admissible first-order change of representative is restricted to the
slowly varying adiabatic sector. Componentwise, its amplitude and velocity
are counted as
\begin{equation}
\frac{\alpha^\Phi}{\Mpl}
=
{\mathcal O}_{\rm ad}(\epsilon),
\qquad
\frac{\alpha^{M_{{\rm H},0}}}{M_{{\rm H},0}}
=
{\mathcal O}_{\rm ad}(\epsilon),
\label{eq:part5_alpha_amplitude_counting}
\end{equation}
and
\begin{equation}
\frac{r_{{\rm H},0}}{\Mpl}\,
\dot\alpha^\Phi
=
{\mathcal O}_{\rm ad}(\epsilon^2),
\qquad
\frac{\dot\alpha^{M_{{\rm H},0}}}{\Mpl^2}
=
{\mathcal O}_{\rm ad}(\epsilon^2).
\label{eq:part5_alpha_velocity_counting}
\end{equation}
The second set of relations ensures that the change of representative
modifies the collective velocity only at second adiabatic order. Thus,
componentwise,
\begin{equation}
\frac{r_{{\rm H},0}}{\Mpl}
\left(
\dot{\widetilde\lambda}^{\Phi}
-
\dot\lambda^\Phi
\right)
=
{\mathcal O}_{\rm ad}(\epsilon^2),
\qquad
\frac{
\dot{\widetilde\lambda}^{M_{{\rm H},0}}
-
\dot\lambda^{M_{{\rm H},0}}
}{
\Mpl^2
}
=
{\mathcal O}_{\rm ad}(\epsilon^2).
\label{eq:part5_shifted_velocity}
\end{equation}
Hence the collective velocity, and therefore the first-order forcing, is
unchanged at the order retained here.

Expanding the static family about the original representative gives
\begin{align}
\Psi_0\bigl(r;\widetilde\lambda(v)\bigr)
={}&
\Psi_0\bigl(r;\lambda(v)\bigr)
+
\alpha^A(v)
\partial_A\Psi_0\bigl(r;\lambda(v)\bigr)
+
{\mathcal O}_{\rm ad}(\epsilon^2).
\label{eq:part5_shifted_instantaneous}
\end{align}
Since the exact physical solution is independent of the representative
used in the decomposition, the lag field transforms as
\begin{equation}
\Psi_{\rm lag}
\longrightarrow
\Psi_{\rm lag}
-
\alpha^A(v)\partial_A\Psi_0,
\qquad
\dot\alpha^A
=
{\mathcal O}_{\rm ad}(\epsilon^2).
\label{eq:part5_lag_ambiguity}
\end{equation}
Thus the lag field is defined only modulo an adiabatically varying
tangent vector to the static solution manifold.

The two tangent components correspond to the static coordinates
\begin{equation}
\lambda^A
=
\bigl(
\Phi,M_{{\rm H},0}
\bigr).
\end{equation}
After promotion,
$M_{{\rm H},0}(\lambda(v))=M_{\rm H}^{\rm inst}(v)$ labels the horizon
mass of the instantaneous static representative. It is not identified
pointwise with the exact horizon mass
$M_{\rm H}^{\rm exact}(v)$ of the full dynamical solution.

As anticipated in
Section~\ref{subsec:part5_promotion}, we fix the residual tangent
freedom in the static horizon-mass direction by the representation
convention
\begin{equation}
\alpha^{M_{{\rm H},0}}(v)\equiv0.
\label{eq:no-MH-tangent-shift}
\end{equation}
This condition is imposed identically on the admissible changes of
representative. It is not a consequence of the adiabatic counting
alone: a shift satisfying
$\dot\alpha^{M_{{\rm H},0}}
={\mathcal O}_{\rm ad}(\epsilon^2)$
could still accumulate to first adiabatic order over the full slow
timescale.

The consistency of this convention can be seen directly from the
position of the exact marginal horizon. Write, to first order,
\begin{equation}
r_{\rm H}^{\rm exact}(v)
=
r_{\rm H}^{\rm inst}(v)
+
\delta r_{\rm H}(v)
+
{\mathcal O}_{\rm ad}(\epsilon^2),
\label{eq:exact-inst-horizon-splitting}
\end{equation}
where
\begin{equation}
r_{\rm H}^{\rm inst}(v)
=
\frac{2M_{\rm H}^{\rm inst}(v)}{\Mpl^2}
\end{equation}
is the horizon radius of the instantaneous static representative.

Writing the mass lag field in the form derived below,
\begin{equation}
\mu(v,r)
=
f_\phi(r;\lambda(v))\eta(v,r)
+
C_\mu^{\rm lag}(v)e^{-\delta_0(r;\lambda(v))},
\label{eq:lag-mass-form-preview}
\end{equation}
and expanding the exact horizon condition
\begin{equation}
\mathcal F
\bigl(
v,r_{\rm H}^{\rm exact}(v)
\bigr)
=
0
\end{equation}
about $r=r_{\rm H}^{\rm inst}$ gives
\begin{equation}
\delta r_{\rm H}
=
\frac{
2C_\mu^{\rm lag}e^{-\delta_{0{\rm H}}}
}{
\Mpl^2
r_{\rm H}^{\rm inst}
\mathcal F_0'
\bigl(
r_{\rm H}^{\rm inst}
\bigr)
}.
\label{eq:lag-horizon-displacement}
\end{equation}
Here
\begin{equation}
\delta_{0{\rm H}}
\equiv
\delta_0
\bigl(
r_{\rm H}^{\rm inst};
\lambda(v)
\bigr).
\end{equation}

Before the convention given in Eq.~\eqref{eq:no-MH-tangent-shift} is imposed,
$C_\mu^{\rm lag}$ is representation dependent. Under a general
tangent shift,
\begin{equation}
C_\mu^{\rm lag}
\longrightarrow
C_\mu^{\rm lag}
-
\alpha^A C_{\mu,A}.
\label{eq:C-lag-transformation-preview}
\end{equation}
In particular,
\begin{equation}
\delta_\alpha C_\mu^{\rm lag}
=
-\alpha^{M_{{\rm H},0}}
C_{\mu,M_{{\rm H},0}}.
\end{equation}
The Canonical Tangent Theorem gives
\begin{equation}
C_{\mu,M_{{\rm H},0}}
=
e^{\delta_{0{\rm H}}}
r_{\rm H}^{\rm inst}
\mathcal F_0'
\bigl(
r_{\rm H}^{\rm inst}
\bigr),
\end{equation}
and hence, at the order retained,
\begin{equation}
\delta_\alpha(\delta r_{\rm H})
=
-\frac{
2\alpha^{M_{{\rm H},0}}
}{\Mpl^2}.
\label{eq:lag-horizon-shift-transformation}
\end{equation}
On the other hand,
\begin{equation}
\delta_\alpha r_{\rm H}^{\rm inst}
=
\frac{
2\alpha^{M_{{\rm H},0}}
}{\Mpl^2}.
\end{equation}
Therefore
\begin{equation}
\delta_\alpha r_{\rm H}^{\rm exact}
=
\delta_\alpha r_{\rm H}^{\rm inst}
+
\delta_\alpha(\delta r_{\rm H})
=
0
\label{eq:exact-horizon-representation-invariance}
\end{equation}
to first order.

A tangent shift in the static horizon-mass direction consequently
changes only the decomposition of the exact horizon radius into an
instantaneous-background contribution and a lag-induced displacement.
It does not change the exact horizon itself. The convention
given in Eq.~\eqref{eq:no-MH-tangent-shift} therefore fixes a representation without
discarding any physical information.

Once this convention has been adopted,
$C_\mu^{\rm lag}$ may be interpreted as measuring, through
Eq.~\eqref{eq:lag-horizon-displacement}, the displacement of the exact
marginal horizon relative to the horizon of the chosen instantaneous
static representative. This relative displacement is
representation dependent; the exact horizon radius is not.

The only remaining tangential ambiguity is then the modulus direction,
\begin{equation}
\Psi_{\rm lag}
\sim
\Psi_{\rm lag}
-
\alpha^\Phi(v)
\left.
\partial_\Phi\Psi_0
\right|_{M_{{\rm H},0}},
\qquad
\dot\alpha^\Phi
=
{\mathcal O}_{\rm ad}(\epsilon^2).
\label{eq:remaining-Phi-ambiguity}
\end{equation}
This ambiguity is not a physical degree of freedom. It expresses the
non-uniqueness of the decomposition of a nearby dynamical
configuration into an instantaneous point on $\mathcal S$ and a
remainder transverse to the selected representative. A supplementary
condition, or \emph{slice}, is required to choose one representative
of the equivalence class
\eqref{eq:remaining-Phi-ambiguity}. The slice is constructed in
Section~\ref{sec:slice}, after the unsliced lag equation has been
derived.

%%%%%%%%%%%%%%%%%%%%%%%%%%%%%%%%%%%%%%%%%%%%%%%%%%%%%%%%%%%%%

\subsection{Reduction to the lag equation}
\label{subsec:part5_lag_equation}

Write
\begin{equation}
\Psi_{\rm lag}
=
(\eta,\xi,\mu)
\label{eq:lag-components}
\end{equation}
for the scalar, lapse, and mass components of the lag field. These
quantities are of first adiabatic order. Their time dependence arises
both from the slowly varying instantaneous background and from the
dynamics of the lag fields themselves.

Substituting Eq.~\eqref{eq:part5_lag_definition} into the exact system
\eqref{eq:system} and retaining first-order terms gives the
linearization about $\Psi_{\rm inst}$, sourced by minus the defects
computed in Section~\ref{subsec:part5_defect}. Derivatives with
respect to $v$ of the slowly varying background coefficients multiply
first-order lag fields and therefore contribute only at
${\mathcal O}_{\rm ad}(\epsilon^2)$ in the slowly forced sector.

The radial constraints remain homogeneous exactly under parameter
promotion, while the first-order mass-balance defect vanishes by
Eq.~\eqref{eq:first-order-mass-defect-zero}. Hence
\begin{align}
\xi_{,r}
&=
\frac{r}{\Mpl^2}\phi_0'\eta_{,r},
\label{eq:lag-xi-constraint}
\\
\mu_{,r}
&=
-\delta_0'\mu
+
\left(
f_\phi'
+
\delta_0'f_\phi
\right)\eta
+
f_\phi\eta_{,r},
\label{eq:lag-mu-r-constraint}
\\
\mu_{,v}
&=
f_\phi\eta_{,v}
+
{\mathcal O}_{\rm ad}(\epsilon^2).
\label{eq:lag-mu-v-constraint}
\end{align}
All background quantities in these equations are evaluated on the
instantaneous static representative $\lambda(v)$.

The lapse perturbation is reconstructed from
\begin{equation}
\xi_{,r}
=
\frac{r}{\Mpl^2}\phi_0'\eta_{,r},
\qquad
\xi(v,r_{\rm m})=0.
\label{eq:lag-xi-reconstruction}
\end{equation}
Using the background identity
\begin{equation}
f_\phi'
+
\delta_0'f_\phi
=
f_V,
\end{equation}
the radial mass constraint integrates at each fixed $v$ to
\begin{equation}
\mu(v,r)
=
f_\phi(r;\lambda(v))\eta(v,r)
+
C_\mu^{\rm lag}(v)
e^{-\delta_0(r;\lambda(v))}.
\label{eq:part5_4_mu_reduced}
\end{equation}
This establishes the form anticipated in
Eq.~\eqref{eq:lag-mass-form-preview}.

Here, $C_\mu^{\rm lag}(v)$ is the integration function associated with
the radial mass constraint. Eq.~\eqref{eq:part5_4_mu_reduced}
determines the radial dependence of $\mu$ once $\eta$ and
$C_\mu^{\rm lag}$ are specified, but the unsliced bulk equations do
not determine the instantaneous value of
$C_\mu^{\rm lag}$.

Within the slowly varying adiabatic sector, the temporal mass
constraint implies
\begin{equation}
\dot C_\mu^{\rm lag}
=
{\mathcal O}_{\rm ad}(\epsilon^2).
\label{eq:C-lag-slow}
\end{equation}
This estimate is consistent with the counting of
Eq.~\eqref{eq:adiabatic-counting-derivatives}, which applies to
derivatives inherited from the prescribed slow time dependence.

Indeed, on the slowly forced branch, the lag field has first-order
dimensionless amplitude,
\begin{equation}
\frac{\eta}{\Mpl}
=
{\mathcal O}_{\rm ad}(\epsilon),
\end{equation}
and inherits its time dependence from the exterior data and collective
coordinates. Consequently,
\begin{equation}
\frac{r_{{\rm H},0}}{\Mpl}\,
\eta_{,v}
=
{\mathcal O}_{\rm ad}(\epsilon^2).
\end{equation}
For example, if
\begin{equation}
\eta_{\rm ad}(v,r)
=
\dot\Phi_{\rm ext}(v)\,
q\bigl(r;\lambda(v)\bigr),
\end{equation}
then
\begin{equation}
\partial_v\eta_{\rm ad}
=
\ddot\Phi_{\rm ext}\,q
+
\dot\Phi_{\rm ext}\,
\dot\lambda^A\partial_Aq,
\end{equation}
and, after multiplication by \(r_{{\rm H},0}/\Mpl\), both terms are of
second adiabatic order.

This counting does not apply to a generic homogeneous perturbation with a
frequency of order the inverse black-hole scale. Such a mode may instead
satisfy
\begin{equation}
\frac{\eta_{\rm hom}}{\Mpl}
=
{\mathcal O}_{\rm ad}(\epsilon),
\qquad
\frac{r_{{\rm H},0}}{\Mpl}\,
\eta_{{\rm hom},v}
=
{\mathcal O}_{\rm ad}(\epsilon),
\end{equation}
and belongs to the full evolution problem rather than to the slowly forced
adiabatic branch isolated below.

The same cancellation mechanism as in Section~\ref{sec:operator_and_geometry} eliminates the lapse
perturbation from the scalar equation. The result is
\begin{equation}
2r^2\eta_{,vr}
+
2r\eta_{,v}
+
L\eta
=
S,
\label{eq:part5_4_reduced}
\end{equation}
where
\begin{equation}
L
=
\partial_r
\left(
e^{\delta_0}r^2\mathcal F_0\partial_r
\right)
-
\mathcal V_{\rm eff}(r)
\end{equation}
is the frozen-background radial operator defined in
Eq.~\eqref{eq:reduced_operator_definition}, evaluated on
$\lambda(v)$, and
\begin{align}
S(v,r)
=%{}&
\frac{
2C_\mu^{\rm lag}(v)
}{\Mpl^2}
\bigl(r\phi_0'\bigr)'
%\nonumber\\
%&
-
2\dot\Phi_{\rm ext}(v)
\left(
r^2\varphi_\Phi'
+
r\varphi_\Phi
\right).
\label{eq:part5_4_source}
\end{align}
The first term is generated by the residual integration function in
the radial mass constraint. The second is the physical forcing
required to cancel the scalar equation defect of the promoted static
family.

The replacement
\begin{equation}
\dot\lambda^\Phi
=
\dot\Phi_{\rm ext}
+
{\mathcal O}_{\rm ad}(\epsilon^2)
\end{equation}
is independent of the admissible representative, because
\begin{equation}
\dot\alpha^\Phi
=
{\mathcal O}_{\rm ad}(\epsilon^2).
\end{equation}

Under a general tangent shift, the lag components transform as
\begin{equation}
\eta
\longrightarrow
\eta-\alpha^A\varphi_A,
\qquad
\mu
\longrightarrow
\mu-\alpha^A m_A.
\end{equation}
Using the tangent mass decomposition
\begin{equation}
m_A
=
f_\phi\varphi_A
+
C_{\mu,A}e^{-\delta_0},
\end{equation}
one obtains
\begin{equation}
C_\mu^{\rm lag}
\longrightarrow
C_\mu^{\rm lag}
-
\alpha^A C_{\mu,A}.
\label{eq:C-lag-general-transformation}
\end{equation}
Before a representative convention is imposed,
$C_\mu^{\rm lag}$ is therefore not representation invariant.

The convention
given in Eq.~\eqref{eq:no-MH-tangent-shift} removes the static horizon-mass
component of this transformation. For the only remaining admissible
shift, the Canonical Tangent Theorem gives
\begin{equation}
C_{\mu,\Phi}=0,
\end{equation}
and hence
\begin{equation}
C_\mu^{\rm lag}
\longrightarrow
C_\mu^{\rm lag}
\label{eq:part5_Cmu_invariant_phi}
\end{equation}
under the surviving $\Phi$-tangential ambiguity. Once the
$M_{{\rm H},0}$ convention has been fixed,
$C_\mu^{\rm lag}$ is therefore invariant under all remaining
admissible changes of representative.

The surviving ambiguity is carried by the homogeneous tangent
component
\begin{equation}
\alpha^\Phi\varphi_\Phi
\end{equation}
of the scalar lag field, not by
$C_\mu^{\rm lag}$. The role of the local slice is to select a unique
representative of the equivalence class given in
Eq.~\eqref{eq:remaining-Phi-ambiguity} and, for that selected
representative, to express the otherwise undetermined integration
function $C_\mu^{\rm lag}$ as a functional of the scalar lag field.
The unsliced operator $L$ itself is independent of this choice.

Eq.~\eqref{eq:part5_4_reduced} is the central result of the
present section: it is a single linear evolution equation for the
scalar lag field, with one source fixed by the prescribed exterior
rolling and one residual contribution inherited from the radial mass
constraint.

The slice must be distinguished from the exterior matching condition.
The slice fixes the decomposition of a given physical solution into
collective-coordinate motion and a lag field. Matching instead
supplies physical boundary data from the exterior solution and thereby
helps determine the physical solution itself. These two operations are
carried out separately in Section~\ref{sec:slice}.

%%%%%%%%%%%%%%%%%%%%%%%%%%%%%%%%%%%%%%%%%%%%%%%%%%%%%%%%%%%%%%%%

\setcounter{section}{4}
\section{Local Slice, Sliced Operator, and Fredholm Solvability}
\label{sec:slice}

The previous section established the unsliced adiabatic evolution
equation. Its source contains a prescribed exterior contribution and
an integration function inherited from the radial mass constraint.
After the representation convention in the
$M_{{\rm H},0}$ direction has been fixed, the remaining ambiguity is
the freedom to shift the instantaneous representative along the
canonical $\Phi$ direction while leaving the exact physical solution
unchanged.

The purpose of the present section is to remove this remaining
ambiguity by constructing a local slice, to derive the resulting
sliced operator, and to formulate the completed radial
boundary-value problem.

%%%%%%%%%%%%%%%%%%%%%%%%%%%%%%%%%%%%%%%%%%%%%%%%%%%%%%%%%%%%%

\subsection{Construction of the slice}
\label{subsec:part5_slice}

We adopt the local slice
\begin{equation}
\mu(v,r_{\rm m})=0.
\label{eq:part5_5_slice}
\end{equation}
This condition selects one representative of the remaining
$\Phi$-tangential equivalence class. It is a condition on the
decomposition of the physical solution and should not be interpreted
as an additional physical boundary condition on the scalar field.

%%%%%%%%%%%%%%%%%%%%%%%%%%%%%%%%%%%%%%%%%%%%%%%%%%%%%%%%%%%%%

\subsubsection{Why the matching normalization does not define a slice}

A slice assigns a unique decomposition
\begin{equation}
\Psi
=
\Psi_{\rm inst}(\lambda(v))
+
\Psi_{\rm lag}
+
{\mathcal O}_{\rm ad}(\epsilon^2)
\end{equation}
to every dynamical configuration sufficiently close to
$\mathcal S$. It does not alter the unsliced lag equation, Eq.~\eqref{eq:part5_4_reduced}; it selects one element of the equivalence
class
\begin{equation}
\Psi_{\rm lag}
\sim
\Psi_{\rm lag}
-
\alpha^\Phi(v)
\left.
\partial_\Phi\Psi_0
\right|_{M_{{\rm H},0}}
\end{equation}
defined in Eq.~\eqref{eq:remaining-Phi-ambiguity}.

The matching normalization imposed in Section~\ref{sec:reduction},
\begin{equation}
\xi(v,r_{\rm m})=0,
\end{equation}
cannot serve as this slice. Under a general tangent shift,
\begin{equation}
\xi
\longrightarrow
\xi-\alpha^A\Delta_A,
\end{equation}
but
\begin{equation}
\Delta_A(r_{\rm m})
=
\partial_A\delta_0(r_{\rm m};\lambda)
=
0
\end{equation}
for every $A$, since
\begin{equation}
\delta_0(r_{\rm m};\lambda)=0
\end{equation}
identically throughout the static family. Hence,
$\xi(v,r_{\rm m})$ is invariant under tangent changes of
representative and cannot distinguish neighboring representatives; it
fixes instead the independent residual reparametrization of the
advanced-time coordinate.

The tangent space of the static family is two-dimensional, with
infinitesimal representative shifts
\begin{equation}
\alpha^A
=
\bigl(
\alpha^\Phi,
\alpha^{M_{{\rm H},0}}
\bigr).
\end{equation}
The second component has already been removed identically by the
representation convention
\begin{equation}
\alpha^{M_{{\rm H},0}}(v)\equiv0
\end{equation}
introduced in Section~\ref{subsec:part5_lag_slice}. As shown there,
this convention leaves the exact horizon radius unchanged: to first
order, the induced changes in $r_{\rm H}^{\rm inst}$ and in the
lag-induced displacement $\delta r_{\rm H}$ cancel.

The only remaining tangential freedom is therefore the slowly varying
function \(\alpha^\Phi(v)\), satisfying
\begin{equation}
\frac{\alpha^\Phi}{\Mpl}
=
{\mathcal O}_{\rm ad}(\epsilon),
\qquad
\frac{r_{{\rm H},0}}{\Mpl}\,
\dot\alpha^\Phi
=
{\mathcal O}_{\rm ad}(\epsilon^2).
\end{equation}
Changing the representative along this direction modifies the collective
velocity only at second adiabatic order and therefore leaves the
first-order exterior forcing unchanged. A single scalar slice can thus fix
this one remaining functional ambiguity.

%%%%%%%%%%%%%%%%%%%%%%%%%%%%%%%%%%%%%%%%%%%%%%%%%%%%%%%%%%%%%

\subsubsection{Consequences of the local slice}

Under the surviving tangential change of representative,
\begin{equation}
\Psi_{\rm lag}
\longrightarrow
\Psi_{\rm lag}
-
\alpha^\Phi(v)
\left.
\partial_\Phi\Psi_0
\right|_{M_{{\rm H},0}},
\end{equation}
the scalar and mass lag fields transform as
\begin{equation}
\eta
\longrightarrow
\eta-\alpha^\Phi(v)\varphi_\Phi,
\qquad
\mu
\longrightarrow
\mu-\alpha^\Phi(v)m_\Phi,
\label{eq:part5_slice_tangent_shift}
\end{equation}
where
\begin{equation}
\varphi_\Phi
=
\left.
\partial_\Phi\phi_0
\right|_{M_{{\rm H},0}},
\qquad
m_\Phi
=
\left.
\partial_\Phi m_0
\right|_{M_{{\rm H},0}}.
\end{equation}
In particular,
\begin{equation}
\mu(v,r_{\rm m})
\longrightarrow
\mu(v,r_{\rm m})
-
\alpha^\Phi(v)m_\Phi(r_{\rm m}).
\end{equation}
Provided
\begin{equation}
m_\Phi(r_{\rm m})\neq0,
\label{eq:part5_5_transversality}
\end{equation}
the condition
\begin{equation}
\mu(v,r_{\rm m})=0
\end{equation}
determines $\alpha^\Phi(v)$ uniquely.

Eq.~\eqref{eq:part5_5_transversality} is the transversality
condition between the local slice and the canonical tangent direction.
The evaluation functional
\begin{equation}
\Psi_{\rm lag}
\longmapsto
\mu(r_{\rm m})
\end{equation}
can distinguish neighboring representatives only if it is sensitive
to motion along
$\left.\partial_\Phi\Psi_0\right|_{M_{{\rm H},0}}$.
If
$m_\Phi(r_{\rm m})=0$, the condition
$\mu(r_{\rm m})=0$ is blind to that tangent direction and cannot
select a unique representative. A different slicing functional must
then be used.

Using Eq.~\eqref{eq:part5_4_mu_reduced} and the matching normalization
\begin{equation}
\delta_0(r_{\rm m})=0,
\end{equation}
the slice condition gives
\begin{equation}
C_\mu^{\rm lag}(v)
=
-f_\phi(r_{\rm m})\eta(v,r_{\rm m}).
\label{eq:part5_5_Cmu_fixed}
\end{equation}
The local slice therefore expresses the radial mass-integration
function as a boundary evaluation of the selected scalar lag
representative.

The Canonical Tangent Theorem gives
\begin{equation}
C_{\mu,\Phi}=0,
\end{equation}
so the tangent mass decomposition reduces to
\begin{equation}
m_\Phi(r)
=
f_\phi(r)\varphi_\Phi(r)
\label{eq:canonical-tangent-mass-identity}
\end{equation}
for every $r$. Consequently,
\begin{equation}
m_\Phi(r_{\rm m})
=
f_\phi(r_{\rm m})
\varphi_\Phi(r_{\rm m}).
\label{eq:part5_6_identity}
\end{equation}
The transversality condition therefore states equivalently that the
canonical tangent changes the matching-radius value of the static
mass profile.

The same identity clarifies the status of
$C_\mu^{\rm lag}$. Under the surviving admissible shift,
\begin{equation}
C_\mu^{\rm lag}
\longrightarrow
C_\mu^{\rm lag}
-
\alpha^\Phi C_{\mu,\Phi}
=
C_\mu^{\rm lag}.
\end{equation}
Thus, after the
$M_{{\rm H},0}$ representation convention has been fixed,
$C_\mu^{\rm lag}$ is invariant under the remaining
$\Phi$-tangential ambiguity. It is not itself the parameter of that
ambiguity.

Eq.~\eqref{eq:part5_5_Cmu_fixed} should therefore be interpreted
as follows. The value of $C_\mu^{\rm lag}$ is not changed by moving
within the surviving equivalence class, but its representation in
terms of the boundary value $\eta(r_{\rm m})$ becomes explicit once
the unique representative satisfying
$\mu(r_{\rm m})=0$ has been selected.

Before the slice is imposed, solving the radial mass constraint
determines $\mu$ only up to the function
$C_\mu^{\rm lag}(v)$. The slice supplies the missing relation between
this function and the scalar lag field. Note, however, that it does not impose an
independent physical boundary condition on the total scalar
configuration.\footnote{One could instead impose an orthogonality slice of the form
\begin{equation}
\langle\psi,\eta\rangle=0
\end{equation}
for a radial pairing adapted to $L$. We do not adopt such a condition
because it presupposes the functional-analytic structure (i.e. the domain,
pairing, and adjoint operator) that the subsequent analysis is meant
to establish. The local condition
\eqref{eq:part5_5_slice} requires no such prior structure.}

The slice is practical rather than canonical. Wherever
Eq.~\eqref{eq:part5_5_transversality} holds, it selects a unique
representative of the physical perturbation. A different admissible
slice would redistribute the same exact solution between the
collective coordinates and the lag field without changing the total
physical configuration.

The transversality condition
\(m_\Phi(r_{\rm m})\neq0\)
is therefore a property of the chosen slice rather than of the physical
solution itself. If, for a particular choice of matching radius or slicing
functional, this condition were to fail accidentally, the present local mass
slice would cease to distinguish neighboring representatives along the
canonical tangent direction. This does not signal a breakdown of the
adiabatic construction; it simply indicates that another transverse local
slice should be chosen. The Fredholm framework developed below depends only
on the existence of a transverse slice, not on the particular realization
adopted here.

%%%%%%%%%%%%%%%%%%%%%%%%%%%%%%%%%%%%%%%%%%%%%%%%%%%%%%%%%%%%%%%%%%%%

\subsubsection{The sliced operator}
\label{subsubsec:part5_6_Lsl}

Define the background source profile
\begin{equation}
K(r)
\equiv
\frac{2}{\Mpl^2}
\bigl(r\phi_0'\bigr)'.
\label{eq:part5-K-definition}
\end{equation}
All background quantities in this subsection are evaluated on the
instantaneous static representative $\lambda(v)$.

Using
\begin{equation}
C_\mu^{\rm lag}(v)
=
-f_\phi(r_{\rm m})\eta(v,r_{\rm m}),
\end{equation}
the source
\eqref{eq:part5_4_source} splits as
\begin{equation}
S
=
-
f_\phi(r_{\rm m})
K(r)\eta(v,r_{\rm m})
+
S_{\rm ext},
\end{equation}
where
\begin{equation}
K(r)
=
\frac{2}{\Mpl^2}
\left(r\phi_0'(r)\right)'.
\end{equation}
The first term is absorbed into the definition of the sliced operator,
whereas the second remains as the genuine inhomogeneous forcing generated by
the prescribed exterior rolling.

Isolating the spatial part defines the sliced operator
\begin{equation}
L_{\rm sl}\eta
\equiv
L\eta
+
f_\phi(r_{\rm m})
K(r)\eta(r_{\rm m}).
\label{eq:part5_6_Lsl_def}
\end{equation}
The sliced evolution equation is therefore
\begin{equation}
2r^2\eta_{,vr}
+
2r\eta_{,v}
+
L_{\rm sl}\eta
=
S_{\rm ext}.
\label{eq:sliced-evolution-equation}
\end{equation}
Using
\begin{equation}
2r^2\eta_{,vr}
+
2r\eta_{,v}
=
2r\,\partial_r
\bigl(r\eta_{,v}\bigr),
\end{equation}
it may equivalently be written as
\begin{equation}
2r\,\partial_r
\bigl(r\eta_{,v}\bigr)
+
L_{\rm sl}\eta
=
S_{\rm ext}(v,r).
\label{eq:part5_6_transport_form}
\end{equation}

The unsliced spatial problem is governed by the local second-order
differential operator $L$. The slice adds the term
\begin{equation}
f_\phi(r_{\rm m})
K(r)\eta(r_{\rm m}),
\end{equation}
which depends on the evaluation functional
\begin{equation}
{\rm ev}_{r_{\rm m}}
:
\eta
\longmapsto
\eta(r_{\rm m}).
\end{equation}
Thus $L_{\rm sl}$ is nonlocal as an operator on radial functions,
although its differential part remains local. In operator notation,
\begin{equation}
L_{\rm sl}
=
L
+
f_\phi(r_{\rm m})
\bigl|K\bigr\rangle
\bigl\langle{\rm ev}_{r_{\rm m}}\bigr|.
\label{eq:part5_6_rank_one}
\end{equation}
The additional term has rank at most one. It is not a Robin boundary
condition, nor any other local relation between $\eta(r_{\rm m})$ and
$\eta_{,r}(r_{\rm m})$; it is a finite-rank modification of the bulk
operator induced by the choice of representative.

A complete operator problem also requires a domain. The differential
expression is considered on the radial interval
\begin{equation}
r_{\rm H}^{\rm inst}(v)
\leq
r
\leq
r_{\rm m},
\end{equation}
where $r_{\rm H}^{\rm inst}$ is the horizon of the instantaneous
static representative. At the inner endpoint, the lag fields must
satisfy the linearized future-horizon regularity condition obtained
by expanding the exact relation
Eq.~\eqref{eq:131} about the instantaneous static horizon. The lapse
normalization
\begin{equation}
\xi(v,r_{\rm m})=0
\end{equation}
and the mass slice
\begin{equation}
\mu(v,r_{\rm m})=0
\end{equation}
have already been imposed. No physical outer condition on
$\eta(r_{\rm m})$ has yet been supplied; that condition belongs to
the exterior matching problem.

The sliced evolution equation and the frozen radial problem should be
distinguished. Eq.~\eqref{eq:part5_6_transport_form} is first
order in $v$ and contains a radial derivative acting on
$\eta_{,v}$. Its complete solution requires time-dependent initial and
boundary data and may contain freely excited homogeneous
perturbations. The Fredholm analysis below concerns instead the
instantaneous radial operator obtained after restricting to the
slowly forced adiabatic sector.

We next examine the action of the slice on the canonical tangent
mode. The tangent identity of Section~\ref{sec:operator_and_geometry} is
\begin{equation}
L\varphi_A
=
C_{\mu,A}K.
\label{eq:part5_6_L_tangent}
\end{equation}
In the adapted static coordinates
\begin{equation}
\bigl(
\Phi,M_{{\rm H},0}
\bigr),
\end{equation}
the Canonical Tangent Theorem gives
\begin{equation}
C_{\mu,\Phi}=0,
\qquad
L\varphi_\Phi=0.
\end{equation}
Applying the sliced operator therefore yields
\begin{equation}
L_{\rm sl}\varphi_\Phi
=
f_\phi(r_{\rm m})
K(r)\varphi_\Phi(r_{\rm m}).
\label{eq:part5_6_Lsl_phi}
\end{equation}

The tangent mass decomposition holds identically for every $r$:
\begin{equation}
m_A(r)
=
f_\phi(r)\varphi_A(r)
+
C_{\mu,A}e^{-\delta_0(r)}.
\end{equation}
Specializing to the canonical direction reproduces
Eq.~\eqref{eq:canonical-tangent-mass-identity}, and evaluation at
$r=r_{\rm m}$ gives
Eq.~\eqref{eq:part5_6_identity}. Consequently,
\begin{equation}
L_{\rm sl}\varphi_\Phi
=
m_\Phi(r_{\rm m})K(r).
\label{eq:Lsl-canonical-tangent-exact}
\end{equation}

If both
\begin{equation}
m_\Phi(r_{\rm m})\neq0
\label{eq:slice-transversality-recalled}
\end{equation}
and
\begin{equation}
K\not\equiv0
\label{eq:K-nondegeneracy}
\end{equation}
hold, then
\begin{equation}
L_{\rm sl}\varphi_\Phi\neq0.
\label{eq:part5_6_phi_displaced}
\end{equation}
Thus the canonical tangent mode
$\varphi_\Phi\in\ker L$ is displaced from the kernel of the sliced
operator. This conclusion requires both the transversality of the
slice and the nontriviality of the finite-rank update. Neither
condition follows from the other.

The displacement of $\varphi_\Phi$ does not imply that
$L_{\rm sl}$ is injective. Since the slice produces a finite-rank
perturbation, the Fredholm index may remain unchanged and the
one-dimensional homogeneous freedom may be represented by a
different generator, denoted below by $\psi_\star$. If
$K\equiv0$, no relocation occurs and
\begin{equation}
L_{\rm sl}=L.
\end{equation}

At a fixed advanced time, the corresponding quasistatic radial
equation is
\begin{equation}
L_{\rm sl}\eta
=
S_{\rm ext}.
\label{eq:part5_6_quasistatic}
\end{equation}
The subsequent Fredholm analysis applies to this frozen radial
problem, completed by horizon regularity and an exterior matching
functional. It does not by itself establish existence or uniqueness
for the full time-dependent initial-boundary-value problem.

%%%%%%%%%%%%%%%%%%%%%%%%%%%%%%%%%%%%%%%%%%%%%%%%%%%%%%%%%%%%%

\subsection{The adiabatic branch}
\label{subsec:part5_7_adiabatic_branch}

The preceding subsections reduced the first-order near-zone equations
to a sliced evolution equation for the lag field. We now restrict the
solution space to the sector driven by the prescribed slowly varying
exterior modulus.

The sliced evolution equation is
\begin{equation}
2r^2\eta_{,vr}
+
2r\eta_{,v}
+
L_{\rm sl}\eta
=
S_{\rm ext},
\label{eq:part5_7_full_sliced}
\end{equation}
with
\begin{equation}
S_{\rm ext}(v,r)
=
-2\dot\Phi_{\rm ext}(v)
\left(
r^2\varphi_\Phi'
+
r\varphi_\Phi
\right).
\label{eq:part5_7_external_source}
\end{equation}
Its general solution may be decomposed schematically as
\begin{equation}
\eta
=
\eta_{\rm ad}
+
\eta_{\rm hom},
\label{eq:part5_7_decomposition}
\end{equation}
where $\eta_{\rm ad}$ is a particular response to the prescribed
source and $\eta_{\rm hom}$ satisfies
\begin{equation}
2r^2(\eta_{\rm hom})_{,vr}
+
2r(\eta_{\rm hom})_{,v}
+
L_{\rm sl}\eta_{\rm hom}
=
0.
\label{eq:part5_7_homogeneous_evolution}
\end{equation}

The homogeneous sector contains free perturbations of the
instantaneous black hole background, including transient and
ringdown-type contributions not determined by the cosmological
forcing. Such perturbations need not evolve on the prescribed slow
timescale and are not required to satisfy the adiabatic counting used
for the forced branch. Their treatment belongs to the complete
initial-boundary-value problem.

For the slowly forced branch, the lag field has first-order
dimensionless amplitude,
\begin{equation}
\frac{\eta_{\rm ad}}{\Mpl}
=
{\mathcal O}_{\rm ad}(\epsilon),
\label{eq:part5_7_eta-amplitude}
\end{equation}
and all of its time dependence is inherited from the slowly evolving
exterior data and collective coordinates. Consequently,
\begin{equation}
\frac{r_{{\rm H},0}}{\Mpl}\,
\partial_v\eta_{\rm ad}
=
{\mathcal O}_{\rm ad}(\epsilon^2).
\label{eq:part5_7_eta_counting}
\end{equation}
For example, a response of the form
\begin{equation}
\eta_{\rm ad}(v,r)
=
\dot\Phi_{\rm ext}(v)\,
q\bigl(r;\lambda(v)\bigr)
\end{equation}
satisfies
\begin{equation}
\partial_v\eta_{\rm ad}
=
\ddot\Phi_{\rm ext}\,q
+
\dot\Phi_{\rm ext}\,
\dot\lambda^A\partial_Aq.
\end{equation}
After multiplication by \(r_{{\rm H},0}/\Mpl\), both terms are of
second adiabatic order.

It follows that the transport terms evaluated on the slowly forced
branch satisfy
\begin{equation}
\frac{1}{\Mpl}
\left[
2r^2(\eta_{\rm ad})_{,vr}
+
2r(\eta_{\rm ad})_{,v}
\right]
=
{\mathcal O}_{\rm ad}(\epsilon^2).
\end{equation}
They therefore do not contribute to the first-order equation. The leading
adiabatic response obeys
\begin{equation}
L_{\rm sl}\eta_{\rm ad}
=
S_{\rm ext}.
\label{eq:part5_7_adiabatic_radial}
\end{equation}
This is not an independent approximation to the spatial operator. It is
the first-order restriction of the sliced evolution equation to the slowly
forced branch.

The reduction
\begin{equation}
\eqref{eq:part5_7_full_sliced}
\quad\longrightarrow\quad
\eqref{eq:part5_7_adiabatic_radial}
\end{equation}
must not be applied to a generic homogeneous perturbation. A free mode with
a frequency of order the inverse black-hole scale may instead satisfy
\begin{equation}
\frac{\eta_{\rm hom}}{\Mpl}
=
{\mathcal O}_{\rm ad}(\epsilon),
\qquad
\frac{r_{{\rm H},0}}{\Mpl}\,
\partial_v\eta_{\rm hom}
=
{\mathcal O}_{\rm ad}(\epsilon),
\end{equation}
so that its time-derivative terms remain of the same perturbative order as
\(L_{\rm sl}\eta_{\rm hom}\).

The radial equation
Eq.~\eqref{eq:part5_7_adiabatic_radial} is second order and requires
two radial conditions. The inner condition is regularity at the
future horizon of the instantaneous static representative,
\begin{equation}
r=r_{\rm H}^{\rm inst}(v).
\end{equation}
More precisely, $\eta_{\rm ad}$ must belong to the horizon-regular
domain obtained by linearizing the exact horizon relation
Eq.~\eqref{eq:131} about
$r_{\rm H}^{\rm inst}$. The second condition is supplied by matching
the near-zone solution to the slowly varying exterior at
$r=r_{\rm m}$.

Without specifying a particular exterior model, we represent the
outer matching condition by a bounded linear functional
\begin{equation}
\mathcal B_{\rm m}\eta_{\rm ad}
=
J_{\rm ext}(v),
\label{eq:part5_7_matching-condition}
\end{equation}
with, for a local Robin-type matching law,
\begin{equation}
\mathcal B_{\rm m}
=
A_{\rm m}\,
{\rm ev}_{r_{\rm m}}
+
B_{\rm m}\,
\partial_r{\rm ev}_{r_{\rm m}}.
\label{eq:part5_7_matching}
\end{equation}
The coefficients $A_{\rm m}$ and $B_{\rm m}$, as well as the datum
$J_{\rm ext}$, are determined by the exterior problem. Dirichlet,
Neumann, and Robin conditions arise as special cases. More general
bounded linear matching functionals may also be admitted by the
functional-analytic construction.

The matching condition must not be conflated with the local slice
\begin{equation}
\mu(v,r_{\rm m})=0.
\end{equation}
The slice selects a representative of
\begin{equation}
\eta
\sim
\eta-\alpha^\Phi\varphi_\Phi
\end{equation}
and, for that representative, gives
\begin{equation}
C_\mu^{\rm lag}
=
-f_\phi(r_{\rm m})\eta(r_{\rm m}).
\end{equation}
Matching instead specifies how the physical near-zone configuration
is connected to the exterior solution. The slice changes only the
decomposition of a fixed physical solution; the matching data help
determine the physical solution itself.

The action of the rank-one term depends on the completed
boundary-value problem. If a homogeneous Dirichlet domain is imposed,
\begin{equation}
\eta(r_{\rm m})=0,
\end{equation}
then the rank-one term vanishes on that homogeneous domain. For
inhomogeneous Dirichlet data, its value is prescribed and can be
transferred to the source. For Neumann or Robin matching,
$\eta(r_{\rm m})$ remains part of the unknown and the finite-rank term
acts nontrivially.

The slice transversality condition
\begin{equation}
m_\Phi(r_{\rm m})\neq0
\label{eq:part5_7_transversality_again}
\end{equation}
is independent of the matching data. It states that the mass
evaluation used to define the local slice is sensitive to the
canonical tangent direction
\begin{equation}
\left.
\partial_\Phi\Psi_0
\right|_{M_{{\rm H},0}}.
\end{equation}
Equivalently, the slice intersects the residual
$\Phi$-tangent orbits transversally.

The physical slowly forced response should ultimately be selected by
a causal prescription: regular ingoing behavior at the future
horizon, retarded response to the exterior forcing, and no
independently prescribed incoming free perturbation.
Within such a retarded construction, freely excited black hole modes are not
specified independently but are fixed by the past history of the forcing and
the chosen initial data. They therefore decouple from the leading
quasistatic response considered here. The present Fredholm analysis is
restricted to this causally selected adiabatic branch and does not address
the full spectrum of freely propagating homogeneous perturbations. Such a
prescription determines a particular solution of the full evolution
problem and fixes the homogeneous contribution appropriate to the
physical history. It is more precise than simply setting
$\eta_{\rm hom}=0$, since switching on the source can itself generate
transient contributions.

Establishing this causal selection requires analysis of the complete
initial-boundary-value problem
Eq.~\eqref{eq:part5_7_full_sliced}. It cannot be derived from the
quasistatic radial equation alone and is not attempted here. The
present construction identifies instead the instantaneous radial
problem obeyed by the leading slowly forced response after separating
the non-adiabatic free sector.

The completed leading radial problem is therefore
\begin{equation}
L_{\rm sl}\eta_{\rm ad}
=
S_{\rm ext},
\qquad
\eta_{\rm ad}
\ \text{regular at }r=r_{\rm H}^{\rm inst},
\qquad
\mathcal B_{\rm m}\eta_{\rm ad}
=
J_{\rm ext}(v).
\label{eq:part5_7_radial_bvp}
\end{equation}
Existence and uniqueness are not automatic because
$L_{\rm sl}$ may possess a nontrivial homogeneous kernel. The
Fredholm analysis of the completed operator is given in the following
subsection.

At second adiabatic order, the terms
\begin{equation}
2r^2(\eta_{\rm ad})_{,vr}
+
2r(\eta_{\rm ad})_{,v}
\end{equation}
computed from the leading solution become a definite source and must
be restored. The quasistatic radial equation is therefore the leading
member of a systematic adiabatic expansion, not a replacement for the
full evolution equation.

%%%%%%%%%%%%%%%%%%%%%%%%%%%%%%%%%%%%%%%%%%%%%%%%%%%%%%%%%%%%% 
%%%%%%%%%%%%%%%%%%%%%%%%%%%%%%%%%%%%%%%%%%%%%%%%%%%%%%%%%%%%%

\subsection{The local slice and solvability theorem}
\label{subsec:part5_local_slice_solvability}

We now complete the leading radial problem at a fixed advanced time
$v$. The canonical tangent direction identified in
Theorem~\ref{thm:canonical_tangent} governs the leading motion of the
instantaneous static representative introduced in
Section~\ref{subsec:part5_promotion}. The defect of this promoted
representative generates the lag field, whose unsliced scalar equation
is Eq.~\eqref{eq:part5_4_reduced}.

The local slice
\begin{equation}
\mu(v,r_{\rm m})=0
\end{equation}
constructed in Section~\ref{subsec:part5_slice} fixes the residual
$\Phi$-tangential ambiguity and yields
\begin{equation}
C_\mu^{\rm lag}(v)
=
-f_\phi(r_{\rm m})\eta(v,r_{\rm m}).
\end{equation}
The corresponding sliced operator is
\begin{equation}
L_{\rm sl}
=
L
+
f_\phi(r_{\rm m})
\bigl|K\bigr\rangle
\bigl\langle{\rm ev}_{r_{\rm m}}\bigr|,
\qquad
K(r)
\equiv
\frac{2}{\Mpl^2}
\bigl(r\phi_0'\bigr)'.
\label{eq:part5_solvability_Lsl_again}
\end{equation}
All background quantities in this subsection are evaluated on the
instantaneous static representative at the fixed value of $v$ under
consideration.

Let
\begin{equation}
\mathcal B_{\rm m}:X\longrightarrow\mathbb R
\end{equation}
be the bounded linear functional encoding the outer matching condition
of Eq.~\eqref{eq:part5_7_matching}. The completed radial
boundary-value operator is
\begin{equation}
\mathcal A
\equiv
\bigl(
L_{\rm sl},
\mathcal B_{\rm m}
\bigr)
:
X
\longrightarrow
Y\times\mathbb R.
\label{eq:part5_physical_operator}
\end{equation}
The leading adiabatic problem is therefore
\begin{equation}
\mathcal A\eta_{\rm ad}
=
\bigl(
S_{\rm ext},
J_{\rm ext}
\bigr).
\label{eq:part5_physical_problem}
\end{equation}

The kernel structure needed below follows from the analytic properties
of the unsliced operator and from the finite-rank modification induced
by the local slice.

%%%%%%%%%%%%%%%%%%%%%%%%%%%%%%%%%%%%%%%%%%%%%%%%%%%%%%%%%%%%%

\begin{lemma}[Kernel of the sliced operator]
\label{lem:sliced_kernel}

Assume that the functional-analytic results of Appendix~\ref{app:operator} apply, so
that
\begin{equation}
L:X\longrightarrow Y
\end{equation}
is a bounded, surjective Fredholm operator of index one, and assume
that the canonical tangent field is a nonzero admissible element of
$X$ satisfying
\begin{equation}
\ker L
=
\operatorname{span}\{\varphi_\Phi\}.
\label{eq:unsliced-kernel-assumption}
\end{equation}
Suppose also that the local slice is transverse to the canonical
tangent direction,
\begin{equation}
m_\Phi(r_{\rm m})\neq0.
\label{eq:part5_lemma_slice_transversality}
\end{equation}

Then
\begin{equation}
L_{\rm sl}:X\longrightarrow Y
\end{equation}
is surjective and Fredholm of index one. Consequently,
\begin{equation}
\dim\ker L_{\rm sl}=1.
\label{eq:part5_kernel_dimension_one}
\end{equation}
There therefore exists a nonzero
$\psi_\star\in X$, unique up to multiplication by a nonzero constant, such
that
\begin{equation}
\ker L_{\rm sl}
=
\operatorname{span}\{\psi_\star\}.
\label{eq:part5_kernel_psi}
\end{equation}

If, in addition, the finite-rank update is nontrivial,
\begin{equation}
u
\equiv
f_\phi(r_{\rm m})K
\not\equiv0,
\label{eq:part5_nontrivial_rank_one}
\end{equation}
then
\begin{equation}
\psi_\star\not\propto\varphi_\Phi.
\end{equation}
Under the transversality condition,
Eq.~\eqref{eq:part5_lemma_slice_transversality},
the condition Eq.~\eqref{eq:part5_nontrivial_rank_one} is equivalent to
$K\not\equiv0$.
\end{lemma}

\begin{proof}

Write
\begin{equation}
L_{\rm sl}
=
L+u\otimes\ell,
\qquad
u
\equiv
f_\phi(r_{\rm m})K,
\qquad
\ell
\equiv
{\rm ev}_{r_{\rm m}}.
\label{eq:Lsl-rank-one-proof}
\end{equation}

\emph{Surjectivity.}
Let $g\in Y$ be arbitrary. Since $L$ is surjective, there exists
$\eta_0\in X$ such that
\begin{equation}
L\eta_0=g.
\end{equation}
Because
\begin{equation}
\ker L
=
\operatorname{span}\{\varphi_\Phi\},
\end{equation}
every function
\begin{equation}
\eta
=
\eta_0+a\varphi_\Phi,
\qquad
a\in\mathbb R,
\end{equation}
also satisfies
\begin{equation}
L\eta=g.
\end{equation}

The canonical tangent mass identity
Eq.~\eqref{eq:part5_6_identity} gives
\begin{equation}
m_\Phi(r_{\rm m})
=
f_\phi(r_{\rm m})
\varphi_\Phi(r_{\rm m})
=
f_\phi(r_{\rm m})
\ell(\varphi_\Phi).
\label{eq:transversality-product}
\end{equation}
The transversality assumption therefore implies separately that
\begin{equation}
f_\phi(r_{\rm m})\neq0,
\qquad
\ell(\varphi_\Phi)
=
\varphi_\Phi(r_{\rm m})
\neq0.
\label{eq:part5_transversality_components}
\end{equation}
We may choose
\begin{equation}
a
=
-\frac{\ell(\eta_0)}
{\ell(\varphi_\Phi)},
\end{equation}
for which
\begin{equation}
\ell(\eta)=0.
\end{equation}
It follows that
\begin{equation}
L_{\rm sl}\eta
=
L\eta+u\,\ell(\eta)
=
g.
\end{equation}
Since $g\in Y$ was arbitrary, $L_{\rm sl}$ is surjective.

The mechanism is simple: the surjectivity of $L$ provides a particular
solution, while the one-dimensional kernel of $L$ supplies exactly the
freedom required to set its evaluation at $r_{\rm m}$ to zero and
thereby eliminate the rank-one correction.

\emph{Fredholm index and kernel dimension.}
Because $\ell\in X^*$ and $u\in Y$, the operator
\begin{equation}
u\otimes\ell:X\longrightarrow Y
\end{equation}
is bounded:
\begin{equation}
\left\|
(u\otimes\ell)\eta
\right\|_Y
=
|\ell(\eta)|\,\|u\|_Y
\leq
\|\ell\|_{X^*}\,
\|u\|_Y\,
\|\eta\|_X.
\end{equation}
Its range is contained in
$\operatorname{span}\{u\}$, and hence it has rank at most one. It is
therefore compact.

Stability of the Fredholm index under compact perturbations gives
\begin{equation}
\operatorname{ind}L_{\rm sl}
=
\operatorname{ind}L
=
1.
\label{eq:part5_index_stability}
\end{equation}
Since $L_{\rm sl}$ is surjective,
\begin{equation}
\dim\operatorname{coker}L_{\rm sl}=0,
\end{equation}
and therefore
\begin{equation}
\dim\ker L_{\rm sl}
=
\operatorname{ind}L_{\rm sl}
=
1.
\end{equation}
This proves Eq.~\eqref{eq:part5_kernel_psi}.

\emph{Displacement of the canonical tangent.}
The Canonical Tangent Theorem gives
\begin{equation}
L\varphi_\Phi=0.
\end{equation}
Consequently,
\begin{align}
L_{\rm sl}\varphi_\Phi
&=
u\,\ell(\varphi_\Phi)
\nonumber\\
&=
f_\phi(r_{\rm m})
K(r)
\varphi_\Phi(r_{\rm m})
\nonumber\\
&=
m_\Phi(r_{\rm m})K(r).
\label{eq:part5_Lsl_phi_exact}
\end{align}
If
\begin{equation}
m_\Phi(r_{\rm m})\neq0,
\qquad
K\not\equiv0,
\end{equation}
then
\begin{equation}
L_{\rm sl}\varphi_\Phi\neq0.
\end{equation}
Thus
\begin{equation}
\varphi_\Phi\notin\ker L_{\rm sl}.
\end{equation}
Since the sliced kernel is one-dimensional, its generator $\psi_\star$
cannot be proportional to $\varphi_\Phi$.
\end{proof}

%%%%%%%%%%%%%%%%%%%%%%%%%%%%%%%%%%%%%%%%%%%%%%%%%%%%%%%%%%%%%

\begin{theorem}[Slice and Solvability Theorem]
\label{thm:solvability}

Assume the hypotheses of Lemma~\ref{lem:sliced_kernel}, and let
$\psi_\star$ generate the one-dimensional kernel of $L_{\rm sl}$. Then
\begin{equation}
\ker\mathcal A=\{0\}
\quad\Longleftrightarrow\quad
\mathcal B_{\rm m}\psi_\star\neq0.
\label{eq:part5_solvability_equivalence}
\end{equation}

The completed operator
\begin{equation}
\mathcal A
=
\bigl(
L_{\rm sl},
\mathcal B_{\rm m}
\bigr)
:
X
\longrightarrow
Y\times\mathbb R
\end{equation}
is Fredholm of index zero. Whenever the equivalent conditions in
Eq.~\eqref{eq:part5_solvability_equivalence} hold, $\mathcal A$ is
bijective. Consequently, for every
\begin{equation}
\bigl(
S_{\rm ext},
J_{\rm ext}
\bigr)
\in
Y\times\mathbb R,
\end{equation}
the completed adiabatic boundary-value problem
\begin{equation}
\mathcal A\eta_{\rm ad}
=
\bigl(
S_{\rm ext},
J_{\rm ext}
\bigr)
\end{equation}
has a unique solution
\begin{equation}
\eta_{\rm ad}\in X.
\end{equation}
\end{theorem}

\begin{proof}

By Lemma~\ref{lem:sliced_kernel},
$L_{\rm sl}$ is surjective and
\begin{equation}
\ker L_{\rm sl}
=
\operatorname{span}\{\psi_\star\}.
\end{equation}
Given $S_{\rm ext}\in Y$, choose any
$\eta_0\in X$ satisfying
\begin{equation}
L_{\rm sl}\eta_0
=
S_{\rm ext}.
\end{equation}
Every solution of this radial equation is of the form
\begin{equation}
\eta
=
\eta_0+t\psi_\star,
\qquad
t\in\mathbb R.
\label{eq:part5_all_sliced_solutions}
\end{equation}
The outer matching condition becomes
\begin{equation}
t\,
\mathcal B_{\rm m}\psi_\star
=
J_{\rm ext}
-
\mathcal B_{\rm m}\eta_0.
\label{eq:part5_matching_scalar_equation}
\end{equation}
For arbitrary data, this scalar equation has a unique solution for
$t$ if and only if
\begin{equation}
\mathcal B_{\rm m}\psi_\star\neq0.
\end{equation}
Under this condition, $\mathcal A$ is both surjective and injective,
and hence bijective.

The kernel of the completed operator is
\begin{align}
\ker\mathcal A
&=
\left\{
\eta\in X:
L_{\rm sl}\eta=0,
\quad
\mathcal B_{\rm m}\eta=0
\right\}
\nonumber\\
&=
\ker L_{\rm sl}
\cap
\ker\mathcal B_{\rm m}.
\end{align}
Since
\begin{equation}
\ker L_{\rm sl}
=
\operatorname{span}\{\psi_\star\},
\end{equation}
this intersection is trivial precisely when
\begin{equation}
\mathcal B_{\rm m}\psi_\star\neq0.
\end{equation}
This proves Eq.~\eqref{eq:part5_solvability_equivalence}.

Finally, adjoining one scalar bounded functional to a Fredholm
operator of index one lowers the Fredholm index by one. Equivalently,
$\mathcal A$ differs by a finite-rank operation from
$L_{\rm sl}$ together with a one-dimensional codomain extension.
Therefore
\begin{equation}
\operatorname{ind}\mathcal A
=
\operatorname{ind}L_{\rm sl}-1
=
0.
\end{equation}
\end{proof}

The local slice displaces the canonical tangent field
$\varphi_\Phi$ from the kernel only when the induced finite-rank
update acts nontrivially. If
\begin{equation}
m_\Phi(r_{\rm m})\neq0,
\qquad
K\not\equiv0,
\end{equation}
then the original zero mode is removed from the sliced kernel. The
Fredholm index and surjectivity nevertheless force a one-dimensional
kernel to remain, generated by a generally different function
$\psi_\star$. The outer matching condition eliminates this remaining
homogeneous freedom precisely when
\begin{equation}
\mathcal B_{\rm m}\psi_\star\neq0.
\end{equation}

%%%%%%%%%%%%%%%%%%%%%%%%%%%%%%%%%%%%%%%%%%%%%%%%%%%%%%%%%%%%%

\begin{remark}[Degenerate finite-rank limit]

The Slice and Solvability Theorem does not require
$K\not\equiv0$; it requires only the hypotheses of
Lemma~\ref{lem:sliced_kernel}, including slice transversality. If
\begin{equation}
K\equiv0,
\end{equation}
then
\begin{equation}
L_{\rm sl}=L,
\end{equation}
and no kernel relocation occurs. In this case,
\begin{equation}
\ker L_{\rm sl}
=
\ker L
=
\operatorname{span}\{\varphi_\Phi\},
\end{equation}
so one may choose
\begin{equation}
\psi_\star=\varphi_\Phi.
\end{equation}
The completed operator becomes
\begin{equation}
\mathcal A
=
\bigl(
L,
\mathcal B_{\rm m}
\bigr),
\end{equation}
and its invertibility is governed by
\begin{equation}
\mathcal B_{\rm m}\varphi_\Phi\neq0.
\end{equation}

This is not a separate solvability mechanism. It is the theorem above
specialized to the case in which the slice changes the decomposition
but does not modify the reduced scalar operator. The additional
assumption $K\not\equiv0$ is needed only to conclude that the sliced
kernel generator differs from the original canonical tangent mode.
\end{remark}

%%%%%%%%%%%%%%%%%%%%%%%%%%%%%%%%%%%%%%%%%%%%%%%%%%%%%%%%%%%%%

\begin{remark}[Role of \texorpdfstring{$\psi_\star$}{psi*}]

The function $\psi_\star$ generates the homogeneous kernel of
$L_{\rm sl}$; it is not itself the forced adiabatic deformation. If
the slice is transverse and $K\not\equiv0$, then
\begin{equation}
\psi_\star\not\propto\varphi_\Phi.
\end{equation}
If $K\equiv0$, one may instead choose
\begin{equation}
\psi_\star=\varphi_\Phi.
\end{equation}

In either case, $\psi_\star$ identifies the unique homogeneous radial
direction that must be fixed by the outer matching condition. Whenever
\begin{equation}
\mathcal B_{\rm m}\psi_\star\neq0,
\end{equation}
the physical leading adiabatic deformation is the unique solution of
the completed boundary-value problem, generated by the exterior source
and matching data through
\begin{equation}
\eta_{\rm ad}
=
\mathcal A^{-1}
\bigl(
S_{\rm ext},
J_{\rm ext}
\bigr).
\end{equation}

The condition
\begin{equation}
\mathcal B_{\rm m}\psi_\star\neq0
\end{equation}
states that the generator of the sliced homogeneous kernel does not
satisfy the homogeneous outer matching condition. It is this
condition, rather than the mere displacement of
$\varphi_\Phi$, that makes the completed boundary-value operator
invertible.
\end{remark}

%%%%%%%%%%%%%%%%%%%%%%%%%%%%%%%%%%%%%%%%%%%%%%%%%%%%%%%%%%%%%

\begin{remark}[Analytic input]

Appendix~\ref{app:operator} supplies the functional-analytic results for the unsliced
operator. It defines the Banach spaces $X$ and $Y$, constructs the
unique horizon-regular solution of
\begin{equation}
L\eta=g
\end{equation}
for prescribed horizon value through a regular Volterra equation, and
shows that
\begin{equation}
L:X\longrightarrow Y
\end{equation}
is surjective, has a one-dimensional kernel, and possesses a bounded
right inverse. In particular, $L$ is Fredholm of index one.

To identify that one-dimensional kernel specifically with the
canonical tangent field, one must additionally verify that
\begin{equation}
\varphi_\Phi\in X,
\qquad
\varphi_\Phi\neq0.
\label{eq:canonical-tangent-functional-admissibility}
\end{equation}
The Canonical Tangent Theorem then gives
\begin{equation}
L\varphi_\Phi=0,
\end{equation}
and the one-dimensionality of the kernel implies
\begin{equation}
\ker L
=
\operatorname{span}\{\varphi_\Phi\}.
\end{equation}

Thus, the analytic input consists of both the Fredholm properties proved
in Appendix~\ref{app:operator} and the functional admissibility of the canonical
tangent profile. Once these are established, the passage from $L$ to
$L_{\rm sl}$ and the matching criterion of
Theorem~\ref{thm:solvability} follow from the finite-rank argument
given above.
\end{remark}

%%%%%%%%%%%%%%%%%%%%%%%%%%%%%%%%%%%%%%%%%%%%%%%%%%%%%%%%%%%%%

\begin{remark}[Three non-degeneracy conditions]

Three logically distinct non-degeneracy conditions enter the
construction.

First,
\begin{equation}
m_\Phi(r_{\rm m})\neq0
\end{equation}
is the transversality condition for the local slice. It guarantees
that
\begin{equation}
\mu(r_{\rm m})=0
\end{equation}
selects a unique representative of the residual
$\Phi$-tangential equivalence class.

Second,
\begin{equation}
K\not\equiv0
\end{equation}
ensures that the finite-rank term induced by the slice is nontrivial.
Together with slice transversality, it guarantees that
$\varphi_\Phi$ is displaced from the kernel of the reduced operator.

Third,
\begin{equation}
\mathcal B_{\rm m}\psi_\star\neq0
\end{equation}
is the outer-matching non-degeneracy condition. It guarantees that the
one-dimensional kernel of the sliced operator does not survive in the
completed homogeneous boundary-value problem.

The first condition concerns the geometry of the decomposition, the
second the action of the slice on the reduced operator, and the third
the physical matching to the exterior. None of these three conditions
implies either of the others.
\end{remark}

%%%%%%%%%%%%%%%%%%%%%%%%%%%%%%%%%%%%%%%%%%%%%%%%%%%%%%%%%%%%%%%%%%%%%%%%%%%%

%%%%%%%%%%%%%%%%%%%%%%%%%%%%%%%%%%%%%%%%%%%%%%%%%%%%%%%%%%%%%%%%%%%%%%%%%%%%

%%%%%%%%%%%%%%%%%%%%%%%%%%%%%%%%%%%%%%%%%%%%%%%%%%%%%%%%%%%%%%%%%%%%%%%%%%%%

\setcounter{section}{5}
\section{Discussion and outlook}
\label{sec:transition_companion}

This paper has developed a general framework for the leading
adiabatic deformation of a static black hole family under slowly
varying exterior scalar data. The construction separates naturally
into a geometric part and a functional-analytic part.

The geometric structure is encoded by the intrinsic covector
\begin{equation}
\mathfrak C_\mu
\in
T_\lambda^*\mathcal S,
\end{equation}
whose components are the residual mass-integration constants obtained
by differentiating the static family. When the static solution
manifold is parameterized locally by the modulus label $\Phi$ and the
static horizon-mass parameter $M_{{\rm H},0}$, the Canonical Tangent
Theorem identifies
\begin{equation}
\ker\mathfrak C_\mu
=
\operatorname{span}
\left\{
\left.
\partial_\Phi
\right|_{M_{{\rm H},0}}
\right\}.
\label{eq:sec5_canonical_direction}
\end{equation}
Thus the tangent direction generated by varying the modulus at fixed
static horizon-mass parameter is the unique static tangent direction
whose residual mass-integration constant vanishes. After promotion,
this canonical tangent controls the leading motion of the
instantaneous static representative driven by the prescribed exterior
modulus.

The second structure arises from the departure of the exact dynamical
solution from the promoted static family. Instantaneous tracking
satisfies the radial constraints and the mass-balance equation through
first adiabatic order, but it fails, in general, to satisfy the scalar
equation. The resulting lag field obeys a reduced scalar equation
containing the radial integration function
$C_\mu^{\rm lag}(v)$.

The representation convention in the
$M_{{\rm H},0}$ direction and the local slice
\begin{equation}
\mu(v,r_{\rm m})=0
\label{eq:sec5_local_slice}
\end{equation}
play distinct roles. The former fixes the decomposition of the exact
horizon radius into the horizon of the instantaneous static
representative and the lag-induced displacement. The latter fixes the
remaining $\Phi$-tangential ambiguity and gives
\begin{equation}
C_\mu^{\rm lag}(v)
=
-f_\phi(r_{\rm m})\eta(v,r_{\rm m}).
\end{equation}
Substitution into the unsliced scalar equation produces the
finite-rank modification
\begin{equation}
L_{\rm sl}
=
L
+
f_\phi(r_{\rm m})
\bigl|K\bigr\rangle
\bigl\langle{\rm ev}_{r_{\rm m}}\bigr|,
\qquad
K(r)
=
\frac{2}{\Mpl^2}
\bigl(r\phi_0'\bigr)'.
\label{eq:sec5_sliced_operator}
\end{equation}

The slice is well defined provided
\begin{equation}
m_\Phi(r_{\rm m})\neq0.
\label{eq:discussion-slice-transversality}
\end{equation}
where
\begin{equation}
   m_\Phi(r)
\equiv
\left.
\partial_\Phi m_0(r)
\right|_{M_{{\rm H},0}} . 
\end{equation}
This is the transversality condition ensuring that the evaluation
functional $\mu(r_{\rm m})$ distinguishes neighboring representatives
along the canonical tangent direction. 
This requirement expresses only the transversality of the chosen local slice.
It is not a physical restriction on the black hole solution itself, and a
different transverse slice would provide an equivalent decomposition of the
same exact spacetime.

The effect of the slice on the reduced operator depends separately on
the background profile $K$. If
\begin{equation}
K\not\equiv0,
\label{eq:discussion-K-nondegeneracy}
\end{equation}
then, together with slice transversality,
\begin{equation}
L_{\rm sl}\varphi_\Phi
=
m_\Phi(r_{\rm m})K
\neq0,
\end{equation}
and the original canonical zero mode
$\varphi_\Phi\in\ker L$ is displaced from the kernel. Stability of the
Fredholm index and surjectivity of the sliced operator nevertheless
leave a one-dimensional homogeneous kernel,
\begin{equation}
\ker L_{\rm sl}
=
\operatorname{span}\{\psi_\star\},
\end{equation}
generated by a generally different profile $\psi_\star$. If
$K\equiv0$, no relocation occurs and one may instead take
$\psi_\star=\varphi_\Phi$.

In either branch, the outer matching functional removes the remaining
homogeneous radial freedom precisely when
\begin{equation}
\mathcal B_{\rm m}\psi_\star\neq0.
\label{eq:sec5_solvability_condition}
\end{equation}
Under slice transversality, this condition is equivalent to
invertibility of the completed boundary-value operator
\begin{equation}
\mathcal A
=
\bigl(
L_{\rm sl},
\mathcal B_{\rm m}
\bigr)
:
X
\longrightarrow
Y\times\mathbb R.
\end{equation}
The leading slowly forced deformation is then the unique solution
\begin{equation}
\eta_{\rm ad}
=
\mathcal A^{-1}
\bigl(
S_{\rm ext},
J_{\rm ext}
\bigr).
\end{equation}

The purpose of the present paper has been to establish this structure
without specializing to a particular black hole family. The argument
requires a smooth two-parameter family of non-extremal static
solutions, the exact spherical reduction, the local scalar operator
constructed from the static background, and the Fredholm properties
proved in Appendix~\ref{app:operator}. It does not determine by itself
the explicit forms of the canonical tangent profile, the sliced
kernel generator, or the exterior matching functional. Those depend
on the particular static family and on the exterior problem to which
the finite near-zone solution is matched.

The companion paper \cite{Benakli:2026PaperII} specializes the framework to the magnetic
GHS family. Because the static profiles are
known explicitly, the static derivatives with respect to
$\Phi$ and $M_{{\rm H},0}$ can be evaluated directly. This makes it
possible to construct
\begin{equation}
\varphi_\Phi
=
\left.
\partial_\Phi\phi_0
\right|_{M_{{\rm H},0}},
\end{equation}
evaluate the slice-transversality factor
$m_\Phi(r_{\rm m})$, determine whether the finite-rank update is
nontrivial, and identify the corresponding generator $\psi_\star$ of the
sliced homogeneous kernel. Once an exterior matching prescription is
specified, the condition
$\mathcal B_{\rm m}\psi_\star\neq0$
can likewise be tested explicitly.

The GHS specialization also provides the setting in which the effects
of a slowly evolving cosmological modulus and of a scalar potential
can be analyzed quantitatively. In particular, it allows one to
determine how the massless GHS scalar profile is modified and to
identify parameter regimes in which the black hole-induced modulus
excursion is preserved, screened, or otherwise deformed. These
model-dependent conclusions are not assumed in the general analysis
presented here.

This division of labor between the two papers is essential. The
Canonical Tangent Theorem, the local slice, the finite-rank
modification of the reduced operator, and the solvability theorem do
not rely on the closed-form GHS solution. Conversely, the companion
paper need not reconstruct the general formalism: its role is to
supply the GHS background data, complete the exterior matching
problem, and extract the resulting physical scalar response.

The main conceptual conclusion may therefore be stated independently
of any particular static family. A slowly evolving exterior modulus
does not determine an adiabatic black hole response merely by replacing
a static modulus parameter by a function of time. The leading physical
response emerges only after four logically distinct operations:

\begin{enumerate}
\item the exact dynamical horizon quantities are distinguished from
the horizon parameters of the instantaneous static representative;

\item the geometrically preferred static tangent direction is
identified through
$\ker\mathfrak C_\mu$;

\item the residual tangential ambiguity is fixed by a slice;

\item the resulting radial problem is completed by physical matching
to the exterior.
\end{enumerate}

For the branch in which the finite-rank update is nontrivial, the
three relevant scalar non-degeneracy conditions are
\begin{equation}
m_\Phi(r_{\rm m})\neq0,
\qquad
K\not\equiv0,
\qquad
\mathcal B_{\rm m}\psi_\star\neq0.
\label{eq:sec5_three_conditions}
\end{equation}
Their roles are distinct. The first ensures that the local slice is transverse, the second ensures relocation of the canonical zero mode, and the third completes the Fredholm criterion: under the Fredholm hypotheses and assuming slice transversality, it is equivalent to invertibility of the completed radial boundary-value problem.

The condition $K\not\equiv0$ is not required for solvability itself.
In the degenerate finite-rank limit $K\equiv0$, one has
\begin{equation}
L_{\rm sl}=L,
\qquad
\psi_\star=\varphi_\Phi,
\end{equation}
and the completed problem is invertible precisely when
\begin{equation}
\mathcal B_{\rm m}\varphi_\Phi\neq0.
\end{equation}
Thus, the universal criterion for existence and uniqueness is that the
generator of the one-dimensional kernel of the sliced operator should
not satisfy the homogeneous outer matching condition,
\begin{equation}
\mathcal B_{\rm m}\psi_{\rm ker}\neq0,
\end{equation}
where $\psi_{\rm ker}$ denotes the generator of
$\ker L_{\rm sl}$. In the generic branch,
$\psi_{\rm ker}=\psi_\star$, whereas in the degenerate branch,
$K\equiv0$, one has
$\psi_{\rm ker}=\varphi_\Phi$.

The companion paper \cite{Benakli:2026PaperII}  determines explicitly which branch is realized by
the charged magnetic GHS family and evaluates these conditions for a
slowly rolling exterior modulus.

\appendix

\section{Analytic properties of the reduced operator}
\label{app:operator}

This appendix establishes the functional-analytic framework underlying
the reduced operator $L$. It provides the analytic input required for
the functional use of the Canonical Tangent Theorem, for
Lemma~\ref{lem:sliced_kernel}, and for Theorem~\ref{thm:solvability}:
namely, that $L$ is a surjective Fredholm operator of index one on a
precisely defined Banach space and has a one-dimensional kernel. Once
the canonical tangent profile $\varphi_\Phi$ is shown to belong to
that domain, Theorem~\ref{thm:canonical_tangent} identifies this
kernel with $\operatorname{span}\{\varphi_\Phi\}$. The construction
proceeds from a regular Volterra problem at the static horizon, from
which uniqueness of the horizon-regular homogeneous solution and
surjectivity of $L$ follow.

%%%%%%%%%%%%%%%%%%%%%%%%%%%%%%%%%%%%%%%%%%%%%%%%%%%%%%%%%%%%%

\subsection{Functional setting}

Let $r_{{\rm H},0}$ and $r_{\rm m}$ denote, respectively, the horizon
radius of the static background and the matching radius. Throughout
this appendix all coefficients of the reduced operator are evaluated
on a fixed static background belonging to the manifold $\mathcal S$.
The associated static horizon-mass parameter
$M_{{\rm H},0}=\Mpl^2r_{{\rm H},0}/2$ plays no explicit role in the
analysis below.

Recall
$L\eta\equiv\dd_r(p\,\dd_r\eta)-\mathcal V_{\rm eff}\,\eta$,
$p(r)=e^{\delta_0}r^2\mathcal F_0(r)$, from Section~\ref{sec:operator_and_geometry},
Eq.~\eqref{eq:reduced_operator_definition}. Define
\begin{equation}
X\equiv\Bigl\{\eta\in C^1\bigl([r_{{\rm H},0},r_{\rm m}]\bigr)\cap C^2\bigl((r_{{\rm H},0},r_{\rm m}]\bigr)\ :\ L\eta\in C^0\bigl([r_{{\rm H},0},r_{\rm m}]\bigr)\Bigr\},
\qquad
Y\equiv C^0\bigl([r_{{\rm H},0},r_{\rm m}]\bigr),
\label{eq:appA_domain}
\end{equation}
equipped with
\begin{equation}
\|g\|_Y\equiv\|g\|_{C^0},
\qquad
\|\eta\|_X\equiv\|\eta\|_{C^1([r_{{\rm H},0},r_{\rm m}])}+\|L\eta\|_{C^0([r_{{\rm H},0},r_{\rm m}])}.
\label{eq:appA_norms}
\end{equation}
With these norms, $Y$ is a Banach space and $L:X\to Y$ is bounded by
construction. The space $X$ is also complete. Indeed, let $(\eta_n)$
be Cauchy in $\|\cdot\|_X$. Then there exist
\begin{equation}
\eta\in C^1\bigl([r_{{\rm H},0},r_{\rm m}]\bigr),
\qquad
g\in C^0\bigl([r_{{\rm H},0},r_{\rm m}]\bigr)
\end{equation}
such that
\begin{equation}
\eta_n\longrightarrow\eta
\quad\text{in }C^1,
\qquad
L\eta_n\longrightarrow g
\quad\text{in }C^0.
\end{equation}
Since $p(r_{{\rm H},0})=0$ and $\eta_n'$ is finite at the horizon,
integration of
\begin{equation}
(p\eta_n')'
=
\mathcal V_{\rm eff}\eta_n+L\eta_n
\end{equation}
gives
\begin{equation}
p(r)\eta_n'(r)
=
\int_{r_{{\rm H},0}}^r
\left[
\mathcal V_{\rm eff}(s)\eta_n(s)
+
L\eta_n(s)
\right]ds.
\end{equation}
Passing to the uniform limit yields
\begin{equation}
p(r)\eta'(r)
=
\int_{r_{{\rm H},0}}^r
\left[
\mathcal V_{\rm eff}(s)\eta(s)+g(s)
\right]ds.
\label{eq:appA-completeness-limit}
\end{equation}
On every compact interval $[r_{{\rm H},0}+\delta,r_{\rm m}]$, with
$\delta>0$, the function $p$ is nonzero.
Eq.~\eqref{eq:appA-completeness-limit} therefore implies
\begin{equation}
\eta\in
C^2\bigl([r_{{\rm H},0}+\delta,r_{\rm m}]\bigr)
\end{equation}
and $L\eta=g$ there. Since $\delta>0$ is arbitrary,
$\eta\in C^2((r_{{\rm H},0},r_{\rm m}])$, and $L\eta=g$ extends
continuously to the horizon by
Eq.~\eqref{eq:appA-completeness-limit}. Hence $\eta\in X$, and
$\|\eta_n-\eta\|_X\to0$. Thus $X$ is a Banach space.

Horizon regularity is built into $X$ directly, through membership in
$C^1([r_{{\rm H},0},r_{\rm m}])$; no condition is imposed at
$r_{\rm m}$, consistent with the matching functional $\mathcal B_{\rm
m}$ of Section~\ref{sec:tracking} acting separately.

%%%%%%%%%%%%%%%%%%%%%%%%%%%%%%%%%%%%%%%%%%%%%%%%%%%%%%%%%%%%%

\subsection{The regular Volterra problem}

Write $x\equiv r-r_{{\rm H},0}$. Since $\mathcal F_0(r_{{\rm H},0})=0$
with $\mathcal F_0'(r_{{\rm H},0})\neq0$ (non-degenerate static
horizon, Section~\ref{sec:operator_and_geometry}), and $e^{\delta_0}$, $r^2$ are smooth and nonzero
at $r_{{\rm H},0}$,
\begin{equation}
p(r)=p_1\,x+O(x^2),
\qquad
p_1\equiv e^{\delta_{0{\rm H}}}r_{{\rm H},0}^2\,\mathcal F_0'(r_{{\rm H},0})\neq0,
\label{eq:appA_p_expansion}
\end{equation}
and $\mathcal V_{\rm eff}$, built from $e^{\delta_0}$, $r^2\mathcal
F_0$, $\phi_0'^2$, $V_0''$, $B_0''$, all smooth at $r_{{\rm H},0}$, is
itself continuous there, with $V_0^{\rm eff}\equiv\mathcal
V_{\rm eff}(r_{{\rm H},0})$ finite. Thus $r=r_{{\rm H},0}$ is a
regular singular point of $L\eta=0$, with $p$ having a simple zero
there.

\begin{proposition}[Existence and uniqueness at fixed horizon data]
\label{prop:appA_volterra}
For every $a\in\mathbb R$ and every $g\in Y$, there exists a unique
$\eta\in X$ with
\begin{equation}
L\eta=g,
\qquad
\eta(r_{{\rm H},0})=a.
\end{equation}
Moreover $\eta'(r_{{\rm H},0})=\bigl(V_0^{\rm eff}a+g(r_{{\rm H},0})\bigr)/p_1$.
\end{proposition}

\begin{proof}
Write $L\eta=g$ as $(p\eta')'=\mathcal V_{\rm eff}\eta+g$. For a
solution with $\eta\in C^1$ at $r_{{\rm H},0}$, $p(r_{{\rm H},0})=0$
and $\eta'(r_{{\rm H},0})$ finite give
$\lim_{r\to r_{{\rm H},0}^+}p(r)\eta'(r)=0$. Integrating once from the
horizon,
\begin{equation}
p(r)\eta'(r)=\int_{r_{{\rm H},0}}^r\bigl[\mathcal V_{\rm eff}(s)\eta(s)+g(s)\bigr]\,ds.
\label{eq:appA_first_integral}
\end{equation}
With $\eta(r_{{\rm H},0})=a$ prescribed, a second integration yields
the Volterra equation
\begin{equation}
\eta(r)=a+\int_{r_{{\rm H},0}}^r\frac{dt}{p(t)}\int_{r_{{\rm H},0}}^t\bigl[\mathcal V_{\rm eff}(s)\eta(s)+g(s)\bigr]\,ds.
\label{eq:appA_volterra}
\end{equation}
For every continuous trial function $\eta$, the quotient
\begin{equation}
\frac{1}{p(t)}
\int_{r_{{\rm H},0}}^t
\left[
\mathcal V_{\rm eff}(s)\eta(s)+g(s)
\right]ds
\end{equation}
extends continuously to $t=r_{{\rm H},0}$ by assigning it the value
\begin{equation}
\frac{
V_0^{\rm eff}a+g(r_{{\rm H},0})
}{p_1},
\end{equation}
using Eq.~\eqref{eq:appA_p_expansion} and continuity of $\mathcal
V_{\rm eff},\eta,g$: the inner integral vanishes linearly in
$(t-r_{{\rm H},0})$ with leading coefficient
$V_0^{\rm eff}a+g(r_{{\rm H},0})$, matching the linear vanishing of
$p(t)$. Consequently, Eq.~\eqref{eq:appA_volterra} defines a regular
linear Volterra equation of the second kind on
$[r_{{\rm H},0},r_{\rm m}]$. Standard successive
approximation~\cite{Coddington:1955} (the Volterra operator is a
contraction on $C^0([r_{{\rm H},0},r_{{\rm H},0}+\delta])$ for
$\delta$ small, and the solution extends to all of
$[r_{{\rm H},0},r_{\rm m}]$ by finitely many such steps, since the
kernel is bounded on compact subintervals) gives a unique solution
$\eta\in C^1([r_{{\rm H},0},r_{\rm m}])$ for every $a\in\mathbb R$,
$g\in Y$. Eq.~\eqref{eq:appA_first_integral} then gives
$\eta\in C^2((r_{{\rm H},0},r_{\rm m}])$ with $L\eta=g$ there, and
$L\eta=g$ extends continuously to $r_{{\rm H},0}$ by construction, so
$\eta\in X$. The stated value of $\eta'(r_{{\rm H},0})$ is the limit
computed above, evaluated at $t=r_{{\rm H},0}$.
\end{proof}

%%%%%%%%%%%%%%%%%%%%%%%%%%%%%%%%%%%%%%%%%%%%%%%%%%%%%%%%%%%%%

\subsection{Consequences}

\begin{corollary}[One-dimensional kernel]
\label{cor:appA_kernel}
$\ker L\cap X$ is one-dimensional.
\end{corollary}

\begin{proof}
Apply Proposition~\ref{prop:appA_volterra} with $g=0$. For each
$a\in\mathbb R$ there is a unique horizon-regular solution of
$L\eta=0$ with $\eta(r_{{\rm H},0})=a$; by uniqueness for fixed $a$,
this solution is $a$ times the one obtained at $a=1$, so the space of
horizon-regular solutions of $L\eta=0$ is exactly the span of the
$a=1$ solution, one-dimensional. \footnote{
Equivalently, if $\eta_1$ and $\eta_2$ satisfy $L\eta=0$, their
Wronskian $W=\eta_1\eta_2'-\eta_2\eta_1'$ obeys $(pW)'=0$ for every
pair of homogeneous solutions, so that $pW\equiv C$. Since
$p(r)\sim p_1(r-r_{{\rm H},0})$ vanishes at the horizon while $W$
remains finite for horizon-regular solutions, one finds $C=0$, and
any two horizon-regular solutions are therefore proportional. This
identity also underlies the Sturm--Liouville structure associated
with $L$, which would be needed for the alternative
orthogonality-based slice briefly mentioned in
Section~\ref{subsec:part5_slice}, but is not pursued in the present
work.
}
\end{proof}

\begin{corollary}[Surjectivity]
\label{cor:appA_surjective}
$L:X\to Y$ is surjective, with bounded right inverse $L^+:Y\to X$
given by $L^+g\equiv\eta_g$, the solution of
Proposition~\ref{prop:appA_volterra} with $a=0$.
\end{corollary}

\begin{proof}
Existence of $\eta_g\in X$ with $L\eta_g=g$ for every $g\in Y$ is
Proposition~\ref{prop:appA_volterra} at $a=0$, giving
$LL^+=\operatorname{id}_Y$. Linearity of $L^+$ follows from linearity
of Eq.~\eqref{eq:appA_volterra} in $g$ at fixed $a=0$. For
boundedness, the simple-zero property of $p$,
Eq.~\eqref{eq:appA_p_expansion}, implies
\begin{equation}
\sup_{r\in(r_{{\rm H},0},r_{\rm m}]}\frac{r-r_{{\rm H},0}}{|p(r)|}<\infty.
\end{equation}
Eq.~\eqref{eq:appA_first_integral} therefore gives
\begin{equation}
|\eta_g'(r)|\leq C\Bigl(\sup_{s\in[r_{{\rm H},0},r]}|\eta_g(s)|+\|g\|_{C^0}\Bigr).
\end{equation}
Since $\eta_g(r_{{\rm H},0})=0$, integration and Gronwall's inequality
yield $\|\eta_g\|_{C^0}\leq C_0\|g\|_{C^0}$ and, substituting back,
$\|\eta_g'\|_{C^0}\leq C_1\|g\|_{C^0}$. Consequently
$\|L^+g\|_X=\|\eta_g\|_{C^1}+\|g\|_{C^0}\leq C\|g\|_Y$.
\end{proof}

\begin{corollary}[Fredholm index]
\label{cor:appA_fredholm}
$L:X\to Y$ is Fredholm of index one:
$\operatorname{ind}L=\dim\ker L-\dim\operatorname{coker}L=1$.
\end{corollary}

\begin{proof}
By Corollary~\ref{cor:appA_surjective}, the range of $L$ is all of
$Y$, hence closed, with $\operatorname{coker}L=\{0\}$. By
Corollary~\ref{cor:appA_kernel}, $\dim\ker L=1$. Thus $L$ is Fredholm
and $\operatorname{ind}L=1-0=1$.
\end{proof}

\bigskip
\noindent
By the regularity assumption on the static solution manifold made in
Section~\ref{subsec:canonical-tangent}, the neighboring static
profiles are defined on a common radial neighborhood of the reference
horizon and depend differentiably on the collective coordinates at
fixed areal radius. We assume this dependence is sufficiently regular
that
\begin{equation}
\varphi_\Phi
\in
C^1\bigl([r_{{\rm H},0},r_{\rm m}]\bigr)
\cap
C^2\bigl((r_{{\rm H},0},r_{\rm m}]\bigr).
\label{eq:appA-canonical-admissibility}
\end{equation}
Since Theorem~\ref{thm:canonical_tangent} gives $L\varphi_\Phi=0$, it
follows that $\varphi_\Phi\in X$. We further assume, as part of the
non-degeneracy of the modulus parameterization, that
$\varphi_\Phi\not\equiv0$. Corollary~\ref{cor:appA_kernel} then
implies
\begin{equation}
\ker L
=
\operatorname{span}\{\varphi_\Phi\}.
\end{equation}

These are precisely the analytic assumptions used in
Lemma~\ref{lem:sliced_kernel}. The bounded right inverse $L^+$ and the
Fredholm-index statement supply the remaining functional-analytic
ingredients required in its proof. The additional assumption
$K\not\equiv0$ entering that lemma concerns the nontriviality of the
finite-rank perturbation itself, rather than any property of the
unsliced operator $L$.

%%%%%%%%%%%%%%%%%%%%%%%%%%%%%%%%%%%%%%%%%%%%%%%%%%%%%%%%%%%%%%%%%%%%%%%%%%%%%

%%%%%%%%%%%%%%%%%%%%%%%%%%%%%%%%%%%%%%%%%%%%%%%%%%%%%%%%%%%%%%%%%%%%%%%%%%%%%%%

%%%%%%%%%%%%%%%%%%%%%%%%%%%%%%%%%%%%%%%%%%%%%%%%%%%%%%%%%%%%%%%%%%%%%%%%%%%%%

\section{Linearization about the static family}
\label{app:linearization}

This appendix derives the linearized field equations used in
Section~\ref{sec:operator_and_geometry}. The derivation follows from a
direct first-order expansion of the exact
Einstein--Maxwell--dilaton system, Eq.~\eqref{eq:system}, about a fixed
static background
\begin{equation}
\bigl(
\phi_0(r),\delta_0(r),m_0(r)
\bigr),
\end{equation}
using the perturbation ansatz
Eq.~\eqref{eq:linear_perturbations}. All background coefficients in
this appendix are evaluated at a fixed point
$\lambda_\star\in\mathcal S$.

The perturbation of the metric function is
\begin{equation}
\Delta\mathcal F
=
-\frac{2\mu}{\Mpl^2r},
\label{eq:appB_deltaF}
\end{equation}
and hence
\begin{equation}
\Delta\mathcal F\,\phi_0'
=
-\frac{2\phi_0'}{\Mpl^2r}\mu,
\end{equation}
which is the combination entering the scalar-flux perturbation below.
As an independent consistency check, the resulting expressions may
also be verified by direct symbolic expansion of the exact equations.

%%%%%%%%%%%%%%%%%%%%%%%%%%%%%%%%%%%%%%%%%%%%%%%%%%%%%%%%%%%%%

\subsection{Linearized constraints}

Linearizing
\begin{equation}
\delta_{,r}
=
\frac{r}{2\Mpl^2}\phi_{,r}^2
\end{equation}
gives, at background order,
\begin{equation}
\delta_0'
=
\frac{r}{2\Mpl^2}\phi_0'^2
\equiv
f_\delta,
\end{equation}
where
\begin{equation}
f_\delta(r)
\equiv
\delta_0'(r).
\end{equation}
At linear order one obtains
\begin{equation}
\xi_{,r}
=
\frac{r}{\Mpl^2}
\phi_0'\eta_{,r},
\label{eq:appB_xi_constraint}
\end{equation}
in agreement with
Eq.~\eqref{eq:part3_xi_constraint}.

For later use, define
\begin{equation}
f_V(r)
\equiv
\frac{r^2}{2}V_0'
+
\frac{Q_m^2}{4r^2}B_0',
\qquad
f_\phi(r)
\equiv
\frac{r^2}{2}
\mathcal F_0\phi_0'.
\label{eq:appB-f-definitions}
\end{equation}

Linearizing
\begin{equation}
m_{,r}
=
\frac{r^2}{4}
\mathcal F\phi_{,r}^2
+
\frac{r^2}{2}V(\phi)
+
\frac{B(\phi)Q_m^2}{4r^2}
\end{equation}
gives
\begin{equation}
\mu_{,r}
=
-f_\delta\mu
+
f_V\eta
+
f_\phi\eta_{,r},
\label{eq:appB_mu_r}
\end{equation}
in agreement with
Eq.~\eqref{eq:mass_radial_abc}.

Similarly, linearizing
\begin{equation}
m_{,v}
=
\frac{r^2}{2}
\left[
e^{-\delta}\phi_{,v}^2
+
\mathcal F\phi_{,v}\phi_{,r}
\right]
\end{equation}
about the static background gives
\begin{equation}
\mu_{,v}
=
f_\phi\eta_{,v},
\label{eq:appB_mu_v}
\end{equation}
in agreement with
Eq.~\eqref{eq:mass_temporal_abc}. The term
$e^{-\delta}\phi_{,v}^2$ is quadratic in the perturbations and
therefore does not contribute at linear order.

%%%%%%%%%%%%%%%%%%%%%%%%%%%%%%%%%%%%%%%%%%%%%%%%%%%%%%%%%%%%%

\subsection{Solution of the mass constraints}

The temporal mass constraint,
Eq.~\eqref{eq:appB_mu_v}, implies
\begin{equation}
\partial_v
\bigl(
\mu-f_\phi\eta
\bigr)
=
0,
\end{equation}
because $f_\phi$ is independent of $v$ on the fixed static
background. Therefore
\begin{equation}
\mu(v,r)
=
f_\phi(r)\eta(v,r)
+
f_\mu(r),
\label{eq:appB_mu_eta_h}
\end{equation}
where $f_\mu$ is independent of $v$.

Substituting this expression into the radial mass constraint,
Eq.~\eqref{eq:appB_mu_r}, gives
\begin{equation}
f_\mu'
+
f_\delta f_\mu
=
\left(
f_V
-
f_\delta f_\phi
-
f_\phi'
\right)\eta.
\end{equation}
The static scalar equation implies the background identity
\begin{equation}
f_\phi'
+
f_\delta f_\phi
=
f_V,
\label{eq:appB_background_identity}
\end{equation}
and hence
\begin{equation}
f_\mu'
=
-f_\delta f_\mu.
\end{equation}
Since $f_\delta=\delta_0'$, this integrates to
\begin{equation}
f_\mu(r)
=
C_\mu e^{-\delta_0(r)},
\end{equation}
where $C_\mu$ is independent of both $v$ and $r$. Thus
\begin{align}
\mu(v,r)
&=
f_\phi(r)\eta(v,r)
+
C_\mu e^{-\delta_0(r)}
\nonumber\\
&=
\frac{r^2}{2}
\mathcal F_0(r)\phi_0'(r)\eta(v,r)
+
C_\mu e^{-\delta_0(r)}.
\label{eq:appB_mu_solution}
\end{align}

The constancy of $C_\mu$ is a property of the linearization about a
fixed static background. After promotion of the background parameters,
the corresponding radial integration datum becomes the slowly varying
function $C_\mu^{\rm lag}(v)$ introduced in
Section~\ref{subsec:part5_lag_equation}.

%%%%%%%%%%%%%%%%%%%%%%%%%%%%%%%%%%%%%%%%%%%%%%%%%%%%%%%%%%%%%

\subsection{Linearized scalar equation}

Define
\begin{equation}
\mathcal S(\phi,r)
\equiv
V'(\phi)
+
\frac{Q_m^2}{2r^4}B'(\phi),
\label{eq:appB_scalar_source}
\end{equation}
and
\begin{align}
\mathcal S_0(r)
&\equiv
\mathcal S\bigl(\phi_0(r),r\bigr)
=
V_0'
+
\frac{Q_m^2}{2r^4}B_0',
\\
\mathcal S_{0,\phi}(r)
&\equiv
\left.
\frac{\partial\mathcal S}{\partial\phi}
\right|_{\phi=\phi_0(r)}
=
V_0''
+
\frac{Q_m^2}{2r^4}B_0''.
\label{eq:appB_scalar_sources}
\end{align}
The notation $\mathcal S_{0,\phi}$ avoids confusing differentiation
with respect to the scalar argument with a radial derivative.

The exact scalar equation is
\begin{equation}
0
=
\partial_v
\left(
r^2\phi_{,r}
\right)
+
\partial_r
\left[
r^2\phi_{,v}
+
e^\delta r^2\mathcal F\phi_{,r}
\right]
-
e^\delta r^2\mathcal S(\phi,r).
\end{equation}

The two terms containing explicit $v$-derivatives contribute
\begin{align}
\partial_v
\left(
r^2\phi_{,r}
\right)^{(1)}
&=
r^2\eta_{,vr},
\\
\partial_r
\left(
r^2\phi_{,v}
\right)^{(1)}
&=
\partial_r
\left(
r^2\eta_{,v}
\right)
=
r^2\eta_{,vr}
+
2r\eta_{,v}.
\end{align}
Their sum is therefore
\begin{equation}
2r^2\eta_{,vr}
+
2r\eta_{,v}.
\end{equation}

Its direct linearization gives
\begin{align}
0
={}&
2r^2\eta_{,vr}
+
2r\eta_{,v}
\nonumber\\
&+
\partial_r
\left[
e^{\delta_0}r^2
\left(
\mathcal F_0\eta_{,r}
+
\mathcal F_0\phi_0'\xi
-
\frac{2\phi_0'}{\Mpl^2r}\mu
\right)
\right]
\nonumber\\
&-
e^{\delta_0}r^2
\left(
\mathcal S_{0,\phi}\eta
+
\mathcal S_0\xi
\right).
\label{eq:appB_scalar_linearized}
\end{align}
Eq.~\eqref{eq:appB_scalar_linearized} is the unreduced
linearized scalar equation used in Section~\ref{sec:operator_and_geometry}. The remainder of this
appendix reduces it explicitly using only the linearized constraints
\eqref{eq:appB_xi_constraint} and
\eqref{eq:appB_mu_solution} together with the static scalar equation.

%%%%%%%%%%%%%%%%%%%%%%%%%%%%%%%%%%%%%%%%%%%%%%%%%%%%%%%%%%%%%

\subsection{Reduction to the reduced scalar equation}

The reduction of
Eq.~\eqref{eq:appB_scalar_linearized} is purely algebraic. We first
eliminate the lapse-dependent terms and then substitute the solution
of the mass constraints.

\paragraph{Cancellation of the
\texorpdfstring{$\xi$}{xi} terms.}

The terms involving $\xi$ are
\begin{equation}
\partial_r
\left(
e^{\delta_0}r^2
\mathcal F_0\phi_0'\xi
\right)
-
e^{\delta_0}r^2\mathcal S_0\xi.
\end{equation}
By the product rule, the first term is
\begin{align}
&
\left[
\partial_r
\left(
e^{\delta_0}r^2
\mathcal F_0\phi_0'
\right)
\right]\xi
\nonumber\\
&\qquad
+
e^{\delta_0}r^2
\mathcal F_0\phi_0'\xi_{,r}.
\end{align}
The static scalar equation gives
\begin{equation}
\partial_r
\left(
e^{\delta_0}r^2
\mathcal F_0\phi_0'
\right)
=
e^{\delta_0}r^2\mathcal S_0.
\end{equation}
The terms proportional to an undifferentiated $\xi$ therefore cancel,
leaving
\begin{equation}
e^{\delta_0}r^2
\mathcal F_0\phi_0'\xi_{,r}.
\end{equation}
Using Eq.~\eqref{eq:appB_xi_constraint}, this becomes
\begin{equation}
e^{\delta_0}
\frac{
r^3\mathcal F_0\phi_0'^2
}{\Mpl^2}
\eta_{,r}.
\label{eq:appB_xi_result}
\end{equation}
Only the radial constraint on $\xi$ is required; its explicit integral
representation is never used.

%%%%%%%%%%%%%%%%%%%%%%%%%%%%%%%%%%%%%%%%%%%%%%%%%%%%%%%%%%%%%

\paragraph{Substitution of \texorpdfstring{$\mu$}{mu}.}

Using Eq.~\eqref{eq:appB_mu_solution}, the mass-dependent term in the
scalar flux is
\begin{align}
-\frac{
2e^{\delta_0}r\phi_0'
}{\Mpl^2}\mu
={}&
-e^{\delta_0}
\frac{
r^3\mathcal F_0\phi_0'^2
}{\Mpl^2}\eta
\nonumber\\
&-
\frac{2C_\mu}{\Mpl^2}
r\phi_0'.
\end{align}
Its radial derivative is
\begin{align}
&
-\partial_r
\left[
e^{\delta_0}
\frac{
r^3\mathcal F_0\phi_0'^2
}{\Mpl^2}
\eta
\right]
\nonumber\\
&\qquad
-
\frac{2C_\mu}{\Mpl^2}
\bigl(r\phi_0'\bigr)'.
\label{eq:appB_mu_result}
\end{align}
Expanding the first term produces
\begin{equation}
-e^{\delta_0}
\frac{
r^3\mathcal F_0\phi_0'^2
}{\Mpl^2}
\eta_{,r},
\end{equation}
which cancels Eq.~\eqref{eq:appB_xi_result} exactly. The remaining
part contributes an undifferentiated term proportional to $\eta$,
together with the source proportional to $C_\mu$.

%%%%%%%%%%%%%%%%%%%%%%%%%%%%%%%%%%%%%%%%%%%%%%%%%%%%%%%%%%%%%

\paragraph{Assembly.}

Collecting the remaining terms gives
\begin{align}
0
={}&
2r^2\eta_{,vr}
+
2r\eta_{,v}
+
\partial_r
\left(
e^{\delta_0}r^2
\mathcal F_0\eta_{,r}
\right)
\nonumber\\
&-
\left[
e^{\delta_0}r^2
\mathcal S_{0,\phi}
+
\left(
e^{\delta_0}
\frac{
r^3\mathcal F_0\phi_0'^2
}{\Mpl^2}
\right)'
\right]\eta
\nonumber\\
&-
\frac{2C_\mu}{\Mpl^2}
\bigl(r\phi_0'\bigr)'.
\label{eq:appB_scalar_reduced}
\end{align}
Using
\begin{equation}
p(r)
=
e^{\delta_0}r^2\mathcal F_0
\end{equation}
and
\begin{equation}
\mathcal V_{\rm eff}(r)
=
e^{\delta_0}r^2\mathcal S_{0,\phi}
+
\left(
e^{\delta_0}
\frac{
r^3\mathcal F_0\phi_0'^2
}{\Mpl^2}
\right)',
\end{equation}
in agreement with
Eq.~\eqref{eq:effective_radial_potential},
Eq.~\eqref{eq:appB_scalar_reduced} becomes
\begin{equation}
2r^2\eta_{,vr}
+
2r\eta_{,v}
+
L\eta
=
\frac{2C_\mu}{\Mpl^2}
\bigl(r\phi_0'\bigr)'.
\end{equation}
This is precisely
Eq.~\eqref{eq:reduced_scalar_compact}.

The reduction is therefore algebraically closed. The coupled
linearized system has been reduced to a single local equation for the
scalar perturbation $\eta$, together with the residual
mass-integration constant $C_\mu$. The derivation uses only the
linearized constraints and the static scalar equation. It requires
neither an explicit radial integration of the lapse constraint nor
any slice or exterior matching condition.

\section*{Acknowledgments}
AC acknowledges the support of the Initiative Physique des Infinis (IPI), a research training programme of Idex SUPER at Sorbonne Universit\'e.

%\bibliographystyle{unsrtnat}
%\bibliography{references}

\bibliographystyle{JHEP}
\bibliography{GHS_Time_Paper}% Produces the bibliography via BibTeX.

\end{document}